\documentclass[oneside,reqno]{amsart}
\usepackage{amsmath, amsthm, mathrsfs, amssymb, amsfonts, mathtools}
\usepackage{graphicx}
\graphicspath{ {Figures/} }
\usepackage{hyperref}
\usepackage{booktabs}
\usepackage[backend=biber, style=phys, doi=false, isbn=false, hyperref=true, backref=false, firstinits=true, abbreviate=true, sorting=nyt]{biblatex}
\bibliography{lib}

\numberwithin{equation}{section}
\mathtoolsset{showonlyrefs}

\usepackage{orcidlink}

\reversemarginpar

\usepackage{lmodern}
\makeatletter
\ifcase \@ptsize \relax
  \newcommand{\miniscule}{\@setfontsize\miniscule{4}{5}}
\or
  \newcommand{\miniscule}{\@setfontsize\miniscule{5}{6}}
\or
  \newcommand{\miniscule}{\@setfontsize\miniscule{5}{6}}
\fi
\makeatother

\newcounter{todocounter}

\def\diff{\mathrm{d}}

\DeclareMathOperator{\const}{const}

\def\1{\mathbb{1}} 

\newcommand{\xh}{y_{H}}
\newcommand{\rh}{r_{H}}
\newcommand{\bstar}{b_{\ast}}
\newcommand{\Phis}{s}
\newcommand{\PhisZero}{\Phis^{(0)}}

\newcommand{\amp}{a}
\newtheorem{lem}{Lemma}

\DeclareMathOperator{\HeunG}{HeunG}

\usepackage[foot]{amsaddr}
\begin{document}

\title[Dynamics of nonlinear scalar field with Robin BC on the SAdS background]{Dynamics of nonlinear scalar field with Robin boundary condition on the Schwarzschild\textendash{}Anti-de Sitter background}

\author{Filip Ficek$^{1,2}$\,\orcidlink{0000-0001-5885-7064}}
\author{Maciej Maliborski$^{1,2,\dagger}$\,\orcidlink{0000-0002-8621-9761}}
\address{$^1$University of Vienna, Faculty of Mathematics, Oskar-Morgenstern-Platz 1, 1090 Vienna, Austria}
\address{$^2$University of Vienna, Gravitational Physics, Boltzmanngasse 5, 1090 Vienna, Austria}
\email[]{filip.ficek@univie.ac.at}
\email[]{maciej.maliborski@univie.ac.at}
\address{$^{\dagger}$Corresponding author}

\begin{abstract}

This work concerns the dynamics of conformal cubic scalar field on a\linebreak Schwarzschild\textendash{}anti-de Sitter background. The main focus is on understanding how it depends on the size of the black hole and the Robin boundary condition. We identify a critical curve in the parameter space that separates regions with distinct asymptotic behaviours. For defocusing nonlinearity, the global attractor undergoes a pitchfork bifurcation, whereas for the focusing case, we identify a region of the phase space where all solutions blow up in finite time. In the course of this study we observe an interplay between black hole geometry, boundary conditions, and the nonlinear dynamics of scalar fields in asymptotically anti-de Sitter spacetime.
	
\end{abstract}

\date{\today}

\maketitle

\tableofcontents

\section{Introduction}
\label{sec:Introduction}

Consider the $3+1$ dimensional Schwarzschild-anti-de Sitter (SAdS) black hole, also known as the Kottler-Birmingham solution of the vacuum Einstein equation with negative cosmological constant $\Lambda<0$ \cite{Kottler.1918, Chruściel.2020Book}. In the Schwarzschild-like coordinate system $(t,r,\theta,\varphi)$ the SAdS metric takes the following form
\begin{equation}
        \label{eq:22.09.22_02}
        g = -V \diff{t}^{2} + V^{-1}\diff{r}^{2} + r^{2}\diff{\Omega}^{2}
        \,,
        \quad
        V(r) = 1 - \frac{2M}{r} + \frac{r^{2}}{\ell^{2}}
        \,,
\end{equation}
where $\diff{\Omega}^{2}=\diff{\theta}^{2} + \sin^{2}{\theta}\diff{\varphi}^{2}$ is the line element on the unit two-sphere, $M>0$ is the black hole (BH) mass and $\ell$ is the length scale parameter, which is related to the cosmological constant $\Lambda = -3/\ell^{2}$. The BH radius $\rh$ is the simple real root of $V(\rh)=0$. It is convenient to rewrite $V$ in terms of $\rh$
\begin{equation}
	\label{eq:22.09.22_03}
	V(r) = 1 + r^{2} -\frac{r_H}{r}\left(1+ r_H^2\right)
	\,,
\end{equation}
where we set the units so that $\ell=1$, and $r_{H}$ becomes a dimensionless parameter.

We study the dynamics of self-interacting spherically symmetric scalar field $\phi=\phi(t,r)$ propagating on the SAdS background \eqref{eq:22.09.22_02}
\begin{equation}
	\label{eq:22.09.22_04}
	\square_{g} \phi - m^{2}\phi - \lambda \phi^{3} = 0
	\,,
	\quad
	\lambda = \pm 1
	\,,	
\end{equation}
where $\lambda=-1$ corresponds to focusing and $\lambda=1$ to defocusing nonlinearity. 
The lack of global hyperbolicity of asymptotically AdS (AAdS) spacetimes, and in particular of the SAdS, necessarily rises the issue of boundary conditions (BC).
In order to have a well-defined problem, one has to prescribe proper data at the conformal infinity $r\rightarrow\infty$, also referred to as Scri. 
A close inspection of solutions to \eqref{eq:22.09.22_04} gives the following asymptotics for large distances
\begin{equation}
	\label{eq:23.11.22_01}
	\phi(t,r) \sim c_{+}(t)r^{-(3/2+\nu)} + c_{-}(t)r^{-(3/2-\nu)}
	\,,
\end{equation}
where $\nu=\sqrt{9/4+m^{2}}$. For the massless $m^{2}=0$ scalar field the Dirichlet boundary condition $c_{-}=0$ is enforced if we require square integrability, thus $\phi(t,r) \sim c_{+}(t)r^{-3}$.
For $m^{2}=-2$, which corresponds to the conformal coupling, and which is above the Breitenlohner-Freedman mass bound \cite{Breitenlohner.1982}, we have
\begin{equation}
	\label{eq:23.11.22_02}
	\phi(t,r) \sim c_{+}(t)r^{-2} + c_{-}(t)r^{-1}
	\,,
\end{equation}
and so there is a freedom in making the problem well defined.
In this work we intend to explore this flexibility.
Thus, we study the equation \eqref{eq:22.09.22_04} with $m^{2}=-2$ subject to the Robin boundary condition
\begin{equation}
	\label{eq:23.11.22_03}
	-c_{+}+bc_{-} = \lim_{r\rightarrow\infty} \left(r^{2}\partial_{r}(r\phi) + b (r\phi) \right) =0
	\,,
\end{equation}
which is a simple one-parameter ($b\in\mathbb{R}$) generalisation of the Dirichlet ($c_{-}=0$ equivalently $b=\infty$) and Neumann ($c_{+}=0$ or $b=0$) conditions.

To desingularize \eqref{eq:22.09.22_02} we introduce the null coordinate $v$ defined by
\begin{equation}
	\label{eq:22.09.22_05}
	\diff{v} = \diff{t} + \frac{\diff{r}}{V}
	\,,
\end{equation}
which brings \eqref{eq:22.09.22_05} into the ingoing Eddington-Finkelstein form
\begin{align}
	\label{eq:22.09.22_06}
	g = -V\diff{v}^{2} + 2\diff{v}\diff{r} + r^{2}\diff{\Omega}^{2}
	\,,
\end{align}
which is manifestly regular at the black hole horizon $r=r_H$. In this coordinate system the conformal ($m^{2}=-2$) wave equation \eqref{eq:22.09.22_04} for rescaled scalar variable and the compactified radial coordinate
\begin{equation}
	\label{eq:23.11.10_03}
	\Phi = r\phi
	\,,
	\quad
	y = \frac{1}{r}
	\,,
\end{equation}
becomes
\begin{equation}
	\label{eq:22.09.22_09}
	2\partial_{y}\partial_{v}\Phi - \partial_{y}\left(y^{2}V\partial_{y}\Phi\right) - \left(y\partial_{y}V + \frac{2}{y^{2}}\right)\Phi + \lambda\Phi^{3} = 0
	\,,
\end{equation}
where, with a slight abuse of notation, we write
\begin{equation}
	\label{eq:22.09.22_10}
	V(y) = 1 + y^{-2} - \frac{y}{\xh}\left(1+\xh^{-2}\right)
	\,,
	\quad
	\xh = \rh^{-1}
	\,.
\end{equation}
Now, $y=\xh$ denotes the location of the horizon while $y=0$ corresponds to the conformal infinity. Note that in the new coordinate system the Robin BC \eqref{eq:23.11.22_03} is
\begin{equation}
	\label{eq:22.09.22_11}
	\left.\left(-\partial_{v}\Phi + \partial_{y}\Phi - b \Phi\right)\right|_{y=0} = 0
	\,.
\end{equation}

It is instructive to look how the boundary condition affects the energy of the solution.
Multiplying \eqref{eq:22.09.22_09} by $\partial_{v}\Phi$ we get the local conservation law
\begin{multline}
	\label{eq:23.10.24_1}
	\partial_{v} \left(\frac{y^{2}}{2}V(y)(\partial_{y}\Phi)^{2} - \frac{1}{2} 
	\left(y\partial_{y}V(y)+\frac{2}{y^{2}}\right)\Phi^{2}
	+ \frac{\lambda}{4}\Phi^{4}
	\right) 
	\\
	= \partial_{y} \left( y^{2}\,V(y)\partial_{y}\Phi\partial_{v}\Phi - (\partial_{v}\Phi)^{2} \right)
	\,.
\end{multline}
Integrating \eqref{eq:23.10.24_1} over $y$ we obtain the energy loss formula
\begin{equation}
	\label{eq:23.10.03_010}
	\frac{\diff{E}}{\diff{v}} = -\left.\left(\partial_{v}\Phi\right)^{2}\right|_{y=\xh}
	+ \left.\left(\partial_{v}\Phi\left(\partial_{v}\Phi-\partial_{y}\Phi\right)\right)\right|_{y=0}
	\,,
\end{equation}
where we define the energy integral as
\begin{equation}
	\label{eq:23.10.03_01}
	E = \int_{0}^{\xh}\left(
	\frac{y^{2}}{2}V(y)\left(\partial_{y}\Phi\right)^{2}
	- \frac{1}{2} 
	\left(y\partial_{y}V(y)+\frac{2}{y^{2}}\right)\Phi^{2}
	+ \frac{\lambda}{4}\Phi^{4}
	\right)\diff{y}
	\,.
\end{equation}
Finally, using \eqref{eq:22.09.22_11} we rewrite \eqref{eq:23.10.03_010} as
\begin{equation}
	\label{eq:23.11.10_02}
	\frac{\diff{E}}{\diff{v}} = -\left.\left(\partial_{v}\Phi\right)^{2}\right|_{y=\xh}
	- \left.\frac{b}{2}\partial_{v}\left(\Phi^{2}\right)\right|_{y=0}
	\,.
\end{equation}
Observe that additionally to the negative energy flux through the horizon (the first term), there could be a positive (or negative) flow through the Scri if one chooses boundary data with $\partial_{v}\Phi|_{y=0}\neq 0$. In other words, dissipation or generation of energy at the boundary is absent only for the Dirichlet BC.

This work is an extension of \cite{Bizoń.2020} to a spacetime containing a black hole. As a first step in analysing the dynamics of asymptotically AdS black hole solutions, we study the cubic wave equation \eqref{eq:22.09.22_04} on a fixed SAdS spacetime. This substantial simplification allows us to comprehensively describe the nonlinear dynamics of scalar waves in a model with two parameters: the size of the black hole $\xh$ and the Robin boundary parameter $b$. In addition, we consider both the focusing and defocusing nonlinearities. We are mainly interested in how the nonlinear evolution changes with $(\xh,b)$. Specifically, if there is a global-in-time existence or if solutions can develop a singularity in finite time.
Besides finding a classical pitchfork bifurcation in the defocusing case, for the opposite sign of the nonlinearity, we discover a region of the phase space $(\xh,b)$ where all solutions blow up in finite time. 
 
 Although the main emphasis of this work is on the nonlinear dynamics, a substantial part of the analysis is devoted to the study of static solutions and their linear stability. This is a prerequisite to studying time-dependent solutions, as some static configurations play an essential role in the dynamics as global attractors or threshold solutions. This information allows us to provide a precise asymptotic description of global-in-time solutions and solutions near the threshold between blowup and dispersion.
  
A study of a scalar field in AAdS spacetimes with horizons has a rather long history. Extensive literature considers the hairy black holes in AdS, see \cite{Winstanley.2003, Harada.2023} and references therein. Much effort went into studying the quasi-normal mode (QNM) spectrum of linear fields on the SAdS background, mainly with Dirichlet boundary conditions, e.g., \cite{Horowitz.1999,Berti.2009, Cardoso.2003, Moss.2002, Warnick.2015}. Crucially, some works noted the presence of the unstable mode for certain values of the Robin parameter \cite{Holzegel.2014pzh, Araneda.2017, Kinoshita.2023}.
The well-posedness for the massive linear and nonlinear wave equation on AAdS with Dirichlet BC at Scri has been proven in \cite{HOLZEGEL.2012, Warnick.2013} and \cite{Guo.2020}, respectively. Dynamical studies of conformal scalar field on SAdS background with Dirichlet and Neumann boundary conditions were already performed in \cite{Chan.1997}. However, up to our knowledge, there are no works concerning the dynamics of the conformal cubic wave equation subject to the Robin boundary condition. In particular, the study of dichotomy between dispersion and blowup in SAdS and the nonlinear instability results are new. This work aims to fill that gap and provide a better understanding of the dynamics of the scalar field on SAdS with more general boundary conditions.

The manuscript is structured as follows. The first part extensively studies static solutions and their linear stability. We start with horizonless case $\xh\rightarrow 0$, i.e. the equation in AdS, and only after that we consider solutions in SAdS. Subsequently, we discuss the extreme case of a black hole filling the whole space $\xh\rightarrow\infty$.
The second part concerns the nonlinear dynamics, where we separately discuss the focusing and defocusing nonlinearities. The last section contains conclusions and future directions.

\section{Static solutions in AdS ($\xh = \infty$)}
\label{sec:RegularCase}

Before analysing the equation \eqref{eq:22.09.22_09} we discuss static solutions and their linear stability in the limiting case of $\xh=\infty$. This allows for a better understanding of the limit $\xh\rightarrow\infty$.

When studying the regular case, obtained formally by setting $\rh=0$ in \eqref{eq:22.09.22_02}-\eqref{eq:22.09.22_03}, it is convenient to use the compactified radial coordinate $x$, defined by $r=\tan{x}$. Then, the AdS metric takes the form \cite{Bizoń.2020}
\begin{equation}
	\label{eq:231021_01}
	g_{\textrm{AdS}} = \frac{1}{\cos^{2}{x}}\left(-\diff{t}^{2} + \diff{x}^{2} + \sin^{2}{x}\,\diff{\Omega}^{2}\right)
	\,,
\end{equation}
(recall $\ell=1$).
For this metric the conformal wave equation \eqref{eq:22.09.22_04} is
\begin{equation}
	\label{eq:231021_02}
	\partial_{t}^{2}\phi = \frac{1}{\tan^{2}{x}}\partial_{x}\left(\tan^{2}{x}\,\partial_{x}\phi\right) + \frac{2}{\cos^{2}{x}}\phi - \frac{\lambda}{\cos^{2}{x}}\phi^{3}
	\,.
\end{equation}
We introduce the rescaled scalar field, c.f. \eqref{eq:23.11.10_03}
\begin{equation}
	\label{eq:231021_05}
	\Phi = \tan{x}\, \phi
	\,,
\end{equation}
then \eqref{eq:231021_02} becomes
\begin{equation}
	\label{eq:231021_06e}
	\partial_{t}^{2}\Phi = \partial_{x}^{2}\Phi - \frac{\lambda}{\sin^{2}x}\Phi^{3}
	\,.
\end{equation}
This form of the wave equation is manifestly regular at $x=\pi/2$, a consequence of the conformal mass.
Smooth solutions at the origin behave as
\begin{equation}
	\label{eq:23.11.14_01}
	\Phi(t,x) = \mathcal{O}\left(x\right)
	\,.
\end{equation}
At $x=\pi/2$ we have the following expansion
\begin{equation}
	\label{eq:23.11.14_02a}
	\Phi(t,x) = \Phi_{1}(t)\left(\frac{\pi}{2}-x\right) + \Phi_{2}(t)\left(\frac{\pi}{2}-x\right)^{2} + \Phi_{3}(t)\left(\frac{\pi}{2}-x\right)^{3} + \cdots
	\,,
\end{equation}
with coefficients $\Phi_{j}(t)$, $j>2$, uniquely determined by the free functions $\Phi_1(t)$ and $\Phi_2(t)$, in agreement with discussion above. Consequently we impose
\begin{equation}
	\label{eq:231021_06}
	\left.\left(\partial_{x}\Phi + b\Phi\right)\right|_{x=\pi/2} = 0
	\,,
\end{equation}
which follows from the change of variables $\tan{x}=1/y$ in \eqref{eq:22.09.22_11}. Note that, this sign convention agrees with \cite{Holzegel.2014pzh} and is opposite to \cite{Bizoń.2020}.

Before continuing let us comment on the consequences of the BC on conservation of energy. Multiplying the equation \eqref{eq:231021_06} by $\partial_{t}\Phi$ and integrating over the radial coordinate we obtain the energy loss formula
\begin{equation}
	\label{eq:23.11.10_04}
	\frac{\diff}{\diff{t}}E = \left.\partial_{x}\Phi\, \partial_{t}\Phi\right|_{x=\pi/2}
	\,,
\end{equation}
where
\begin{equation}
	\label{eq:23.11.10_05}
	E = \int_{0}^{\pi/2}\left(\frac{1}{2}\left(\partial_{t}\Phi\right)^{2} + \frac{1}{2}\left(\partial_{x}\Phi\right)^{2} + \frac{\lambda}{4\sin^{2}{x}}\Phi^{4}\right)\diff{x}
	\,,
\end{equation}
is the energy integral. Alternatively, using \eqref{eq:231021_06} we can rewrite \eqref{eq:23.11.10_04} as
\begin{equation}
	\label{eq:23.11.10_06}
	\frac{\diff}{\diff{t}}E = -\left.\frac{b}{2}\partial_{t}\left(\Phi^{2}\right)\right|_{x=\pi/2}
	\,.
\end{equation}
Therefore, except for the Dirichlet and Neumann BC the energy of the solution is not conserved due to the boundary contribution.
Note, that by a simple rewriting of \eqref{eq:23.11.10_06} we can get a conserved quantity $\widetilde{E}$
\begin{equation}
	\label{eq:23.11.10_07}
	\widetilde{E} \equiv E + \left.\frac{b}{2}\left(\Phi^{2}\right)\right|_{x=\pi/2}
	\,,
	\quad \frac{\diff}{\diff{t}}\widetilde{E} = 0
	\,,
\end{equation}
which can be interpreted as the total energy of the system: energy in the 'bulk' plus the energy concentrated at the boundary.

We focus here on static solutions $\Phi(t,x)=\Phis(x)$ that satisfy 
\begin{equation}
	\label{eq:08.11.23_1}
	\Phis'' - \frac{\lambda}{\sin^{2}x}\Phis^{3}=0
	\,,
\end{equation}
and \eqref{eq:231021_06} at $x=\pi/2$. The regularity condition \eqref{eq:23.11.14_01} imposed on $\Phis$ is simply $\Phis(0)=0$.

\subsection{Linear equation}
\label{sec:RegularCaseLinearEquation}

The linear stability analysis of the conformal scalar field on the AdS background subject to the Robin BC was discussed in \cite{Bizoń.2020}. For the reader's convenience here we reproduce the main parts of this analysis.
Substituting
\begin{equation}
	\label{eq:23.11.13_01}
	\Phi(t,x) = e^{i\omega t}\psi(x)
	\,,
\end{equation}
into \eqref{eq:231021_06e} and neglecting the nonlinear term we obtain the eigenvalue problem
\begin{equation}
	\label{eq:23.11.13_02}
	-\omega^{2}\psi = \psi''
	\,.
\end{equation}
For $\omega^{2}>0$ the regular solution is $\psi(x) = C\sin{\omega x}$, $C=\const$. The boundary condition \eqref{eq:231021_06} introduces quantisation for the eigenfrequencies
\begin{equation}
	\label{eq:23.11.13_03}
	\omega = -b\tan\left(\omega \frac{\pi}{2}\right)
	\,.
\end{equation}
An elementary analysis shows that out of the infinitely many non-negative solutions to \eqref{eq:23.11.13_03} $\omega_{n}^{2}$, $n=0,1,\ldots$, the lowest eigenvalue $\omega_{0}^{2}$ vanishes at $b=\bstar\equiv -2/\pi$, and becomes an exponentially growing mode with the exponent $\sqrt{-\omega_{0}^{2}}$ satisfying
\begin{equation}
	\label{eq:23.11.13_04}
	\sqrt{-\omega_{0}^{2}} = -b \tanh\left(\sqrt{-\omega_{0}^{2}}\frac{\pi}{2}\right)
	\,,
\end{equation}
for $b<\bstar$. Therefore, for $b>\bstar$ the trivial solution is linearly stable, while for $b<\bstar$ the single growing mode renders it linearly unstable.

\subsection{Focusing case}
\label{sec:RegularCaseFocusingCase}

\begin{figure}[t]
	\centering
	\includegraphics[width=1\textwidth]{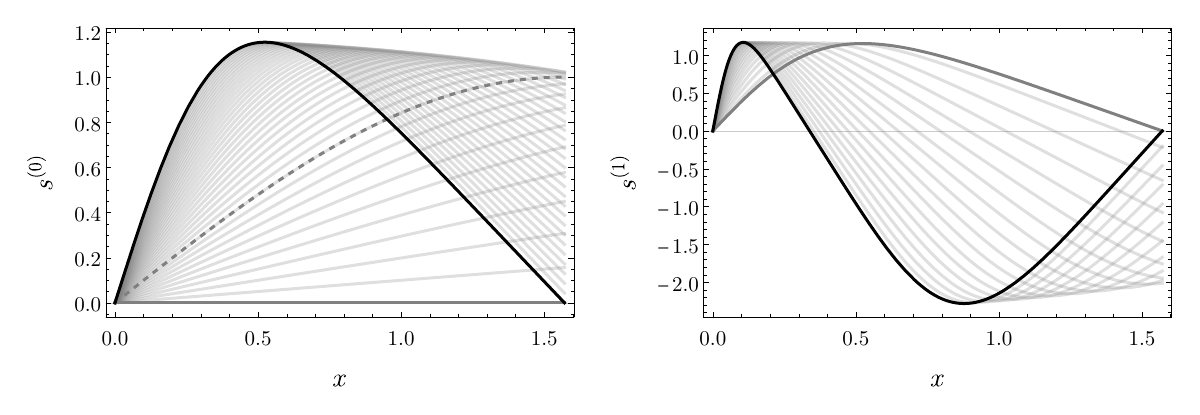}
	\caption{Profiles of the lowest static solutions $\Phis^{(0)}(x)$ and $\Phis^{(1)}(x)$ satisfying Robin BC for the focusing case (light grey). Left plot: the node-less solution bifurcates from zero (dark grey) for $b=\bstar$ and as $b$ increases to infinity it converges to \eqref{eq:23.11.06_01} (black). For $b=0$ this solution is given by \eqref{eq:23.11.16_29} (dashed). Right plot: as $b$ increases from $-\infty$ to $\infty$ the first excited state $\Phis^{(1)}$ smoothly interpolates between two static solutions satisfying Dirichlet BC (dark grey and black for $b=-\infty$ and $b=\infty$ respectively).}
	\label{fig:StaticConformalFocusingGlobal}
\end{figure}

To find static solutions we solve the two point boundary value problem \eqref{eq:08.11.23_1} using shooting method. For a given slope at the origin $s'(0)$ we integrate \eqref{eq:08.11.23_1} outward. For small $s'(0)$ the solution is monotonic and it smoothly extends beyond $x=\pi/2$, a regular point of the equation. For larger values of the initial slope the solution crosses zero finitely many times before it reaches the conformal boundary. Adjusting $s'(0)$ such that the boundary condition \eqref{eq:231021_06} is satisfied, for a given $b$, provides a criterion which selects a particular solution. In this way, we find countably many solutions, $s^{(n)}$, where $n$ is the nodal index, which enumerates number of nodes in the solution, see Fig.~\ref{fig:StaticConformalFocusingGlobal}. The solution $s^{(0)}$ does not exist for $b<\bstar$.

For small $s'(0)$ a node-less solution (with $b>\bstar$) can be constructed perturbatively. We write
\begin{equation}
	\label{eq:23.11.14_02}
	\Phis^{(0)}(x) = \varepsilon f_{1}(x) + \varepsilon^{3} f_{3}(x) + \ldots
	\,,
	\quad
	0<|\varepsilon|\ll 1
	\,,
\end{equation}
where $\Phis^{(0)}{}'(0)=\varepsilon$ in addition to the regularity condition $\Phis^{(0)}(0)=0$. Plugging this into \eqref{eq:08.11.23_1} and expanding in $\varepsilon$ we obtain perturbative equations, which can be solved order by order. The leading equations are
\begin{equation}
	f_1''(x) = 0
	\,,
	\quad
	f_3''(x) = -\frac{f_{1}^{3}}{\sin^{2}{x}}
	\,.
\end{equation}
\begin{figure}[t]
	\centering
	\includegraphics[width=1\textwidth]{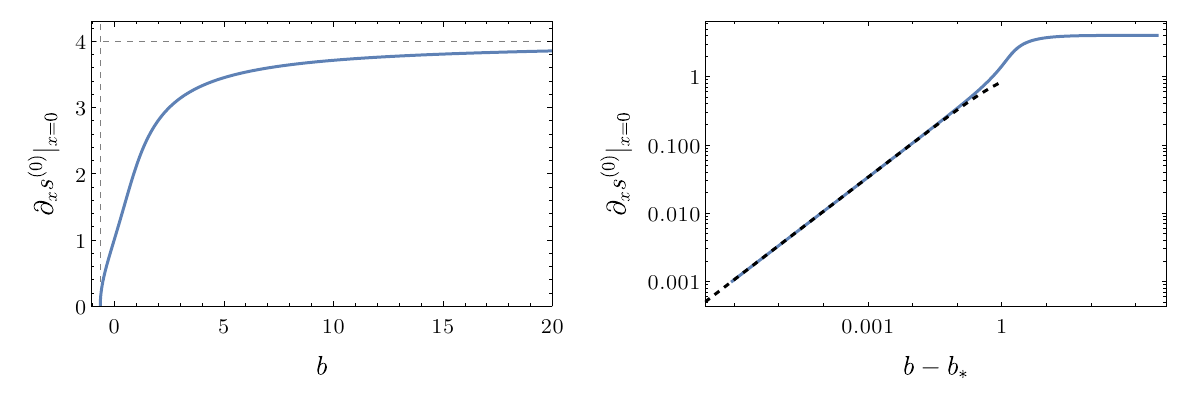}
	\caption{The slope at the origin of the static solution $\PhisZero$ for the focusing case as a function of the Robin parameter $b$. This solution exists only for $b>\bstar$. When $b\rightarrow\infty$ the slope asymptotes the value $4$, the slope of the solution \eqref{eq:23.11.06_01}. On the right plot the dashed line is the perturbative expansion \eqref{eq:23.11.14_02}-\eqref{eq:23.11.14_05}, which agrees with numerical data for $b-\bstar\ll 1$.}
	\label{fig:RegularStaticConformalFocusingPhiOrigin}
\end{figure}
Those are solved by
\begin{equation}
	\label{eq:23.11.14_03}
	f_1(x) = x
	\,,
\end{equation}
and
\begin{multline}
	\label{eq:23.11.14_04}
	f_3(x) =3 x^2 \left(\sum _{k=1}^{\infty } \frac{\sin (2 k x)}{k^2}\right) +\frac{3}{2} x \left(\zeta (3)+3 \sum _{k=1}^{\infty } \frac{\cos (2 k x)}{k^3}\right)\\
	-3 \sum _{k=1}^{\infty } \frac{\sin (2 k x)}{k^4}+x^3 \log (2 \sin x)\, ,
\end{multline}
which alternatively be expressed as
\begin{multline}
	\label{eq:23.11.14_05}
	f_{3}(x) = 3 i x^2 \text{Li}_2\left(e^{-2 i x}\right)+\frac{3}{2} x \left(\zeta (3)+3
   \text{Li}_3\left(e^{-2 i x}\right)\right)
   \\
   -3 i \text{Li}_4\left(e^{-2 i x}\right)+x^3
   \log \left(1-e^{-2 i x}\right)+\frac{i \pi ^4}{30}
	\,,
\end{multline}
where $\text{Li}_n(z)$ is the polylogarithm function \cite{NIST} and $\zeta(z)$ is the Riemann zeta function \cite{NIST}. Note that from this calculation immediately follows that $\bstar=-2/\pi$. Having this perturbative solution one can express $\varepsilon$ by the distance to the bifurcation point, $b-\bstar$, and compare with the numerical data. Results of this test are presented in Fig.~\ref{fig:RegularStaticConformalFocusingPhiOrigin}, where we plot the slope of $\PhisZero$ at the origin as a function of the Robin parameter.

To examine the linear stability of static solutions we substitute the ansatz 
\begin{equation}
	\label{eq:23.11.18_01}
	\Phi(t,x) = \Phis(x) + e^{i\sigma t}\chi(x)
	\,,
\end{equation}
into \eqref{eq:231021_06e} and linearise around $\Phis(x)$. This yields the eigenvalue problem
\begin{equation}
	\label{eq:23.11.18_02}
	-\sigma^{2}\chi = \chi'' + 3\frac{s^{2}}{\sin^{2}{x}}\chi
	\,,
\end{equation}
and the boundary condition
\begin{equation}
	\label{eq:23.11.18_02b}
	\chi'(\pi/2) + b\chi(\pi/2) = 0
	\,,
\end{equation}
which we also solve using the shooting technique.
Starting with $(\chi(0),\chi'(0))=(0,1)$ and some $\sigma^{2}$ we integrate \eqref{eq:23.11.18_02} outward and read off the solution at $x=\pi/2$ (as there are no singularities within the integration domain the solution remains smooth for at least $x\leq \pi/2$).
Adjusting $\sigma^{2}$ so that the condition \eqref{eq:23.11.18_02b} is satisfied we find the desired solution.

Using this numerical procedure we find for $s^{(n)}$ exactly $n+1$ negative eigenmodes\footnote{We use the convention where the super index refers to the nodal index of the static solution and subindex numbers the mode. For unstable modes we use non-positive indices.}
\begin{equation}
	\label{eq:23.11.18_20}
	\left(\sigma^{(n)}_{-n}\right)^{2} < \left(\sigma^{(n)}_{-(n-1)}\right)^{2} < \cdots < \left(\sigma^{(n)}_{0}\right)^{2} < 0
	\,,
\end{equation}
and infinitely many positive modes. Thus, all static solutions $s^{(n)}$ are linearly unstable. 

Additionally, the spectrum of linear perturbations of $\PhisZero$ close to $\bstar$ can be determined perturbatively using the expansion \cite{Bizoń.2020}. The main idea of this rather lengthy calculation, which we skip and refer the reader to \cite{Bizoń.2020}, is to use \eqref{eq:23.11.14_02}-\eqref{eq:23.11.14_05} and a similar expansion for the perturbation $\chi$ at $b=\bstar$ in order to find a correction to the linear spectrum $\omega_{*}\equiv \omega(\bstar)$ \eqref{eq:23.11.13_03} from which $\left(\sigma^{(0)}_{j}\right)^{2}$ bifurcate. For the lowest eigenvalue we obtain
\begin{equation}
	\label{eq:23.11.18_03}
	\left(\sigma^{(0)}_{0}\right)^{2} \approx -\frac{12}{\pi}\delta
	\,,
	\quad
	b=\bstar+\delta
	\,,
	\quad
	0<\delta\ll 1
	\,,
\end{equation}
(for higher modes the expressions are much longer so we omit them). These results compare very well with the data presented Tab.~\ref{tab:RegularQNF}, with the difference decreasing when $b\rightarrow\bstar$. The dependence of the lowest eigenvalues $\left(\sigma^{(0)}_{j}\right)^{2}$ on the boundary parameter is also shown in Fig.~\ref{fig:RegularStaticConformalFocusingEigenfrequencies}, which illustrates the behaviour of the spectra for $b$ close to $\bstar$ and shows the convergence of the eigenspectrum for $b\rightarrow\infty$ to the problem with Dirichlet BC.

\begin{figure}[t]
	\centering
	\includegraphics[width=0.5\textwidth]{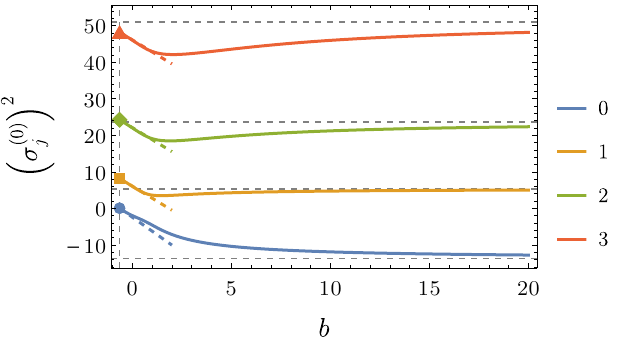}
	\caption{The dependence of the lowest eigenvalues (colour coded) of $\Phis^{(0)}$ on the Robin parameter $b$ for the focusing case. The eigenvalues bifurcate from the linear spectrum \eqref{eq:23.11.13_03} at $b=\bstar$ and for $b-\bstar\ll 1$ they are well approximated by the perturbative expansion (dashed lines). When $b\rightarrow\infty$ the spectrum converges to the eigenvalues of \eqref{eq:23.11.07_01} satisfying Dirichlet BC (horizontal dashed lines).}
	\label{fig:RegularStaticConformalFocusingEigenfrequencies}
\end{figure}

Below, we analyse two special cases for which the fundamental static solutions are explicit. This provides a cross-check of our numerical procedure for finding static solutions and the spectrum of linear perturbations.
We start by considering the Dirichlet BC, which corresponds to taking the limit $b\rightarrow\infty$ in \eqref{eq:231021_06}. We find that the node-less solution takes a simple form
\begin{equation}
	\label{eq:23.11.06_01}
	\PhisZero(x) = \frac{2\sin{2x}}{2-\cos{2x}}
	\,.
\end{equation}
This profile is included in Fig.~\ref{fig:StaticConformalFocusingGlobal} (left plot, solid black line). 
The eigenvalues of \eqref{eq:23.11.06_01} can be computed using Leaver's-type method \cite{Leaver.1985}. In this case linear perturbations $\chi$ \eqref{eq:23.11.18_01} satisfy
\begin{equation}
	\label{eq:23.11.07_01}
	-\sigma^2\chi = \chi''(x)+48\left(\frac{\cos{x}}{2-\cos{2x}}\right)^{2}\chi(x)
	\,.
\end{equation}
Since poles of this equation are located at $\pm 0.658479...\,i +n\pi$, $n\in\mathbb{Z}$, the power series at $x=\pi/2$ will be convergent at the origin. Imposing the condition $\chi(0)=0$ on truncations of such expansion gives us polynomials in $\sigma^{2}$, roots of which correspond to the eigenvalues. The results of these calculations are listed in Tab.~\ref{tab:RegularQNF}. They agree with the numbers obtained from the shooting method.

For Neumann BC ($b=0$) the fundamental solution is
\begin{equation}
	\label{eq:23.11.16_29}
	\PhisZero(x)=\sin{x}
	\,.
\end{equation}
In that case the eigenvalue problem \eqref{eq:23.11.18_02} yields
\begin{equation}
	\label{eq:23.11.16_30}
	\left(\sigma^{(0)}_{0}\right)^{2} = -2\,,
	\quad
	\left(\sigma^{(0)}_{j}\right)^{2} = (2j+1)^{2} - 3
	\,,
\end{equation}
$j=1,2,\ldots$, where the corresponding eigenfunctions are: $\sin{x},\, \sin{3x},\, \sin{5x}\,,\dots$. Thus, we explicitly verify the linear instability of $\PhisZero$ with a single unstable direction.

\begin{table}[t]
\centering
\begin{tabular}{|c|rrrrrr|}
\toprule
$b$& $\left(\sigma^{(0)}_{0}\right)^{2}$ & $\left(\sigma^{(0)}_{1}\right)^{2}$ & $\left(\sigma^{(0)}_{2}\right)^{2}$ & $\left(\sigma^{(0)}_{3}\right)^{2}$ & $\left(\sigma^{(0)}_{4}\right)^{2}$ & $\left(\sigma^{(0)}_{5}\right)^{2}$ \\
\midrule
$\infty$ & $-13.7711$ & $5.29633$ & $23.6731$ & $51.0478$ & $86.7383$ & $130.563$ \\
20 & $-12.8183$ & 4.97254 & 22.2906 & 48.1070 & 81.8407 & 123.373 \\
 10 & $-11.9496$ & 4.68576 & 21.1624 & 45.9244 & 78.5630 & 119.049 \\
 1 & $-4.59103$ & 3.58044 & 19.0428 & 42.9008 & 74.8435 & 114.815 \\
 0 & $-2$ & 6 & 22 & 46 & 78 & 118 \\
 $-1/2$ & $-0.501391$ & 7.72719 & 23.7326 & 47.7341 & 79.7348 & 119.735 \\
 $-2/\pi+10^{-2}$ & $-0.0380887$ & 8.14991 & 24.1542 & 48.1553 & 80.1558 & 120.156 \\
  $-2/\pi+10^{-3}$ &  $-0.00381863$ & 8.17969 & 24.1839 & 48.1850 & 80.1855 & 120.186 \\
  $-2/\pi+10^{-4}$ & $-0.000381961$ & 8.18266 & 24.1869 & 48.1880 & 80.1884 & 120.189 \\
  $-2/\pi$ &  0 & 8.18299 & 24.1872 & 48.1883 & 80.1888 & 120.189 \\
\bottomrule
\end{tabular}
\vskip 2ex
\caption{Lowest eigenvalues of linear perturbation of the static solution $\PhisZero$ for some values of the Robin parameter $b$. Note, the static solution exists only for $b\geq \bstar$. The $b=\infty$ case corresponds to the Dirichlet BC. For the Neumann BC, $b=0$, the spectrum is known explicitly \eqref{eq:23.11.16_30}. At $b=\bstar$ the static solution becomes zero and the eigenvalues are solutions of \eqref{eq:23.11.13_03}. For $b-\bstar\ll 1$ the spectrum can be computed using the perturbative expansion (see the text).}
\label{tab:RegularQNF}
\end{table}

\subsection{Defocusing case}
\label{sec:RegularCaseDefocusingCase}

\begin{figure}[t]
	\centering
	\includegraphics[width=1.0\textwidth]{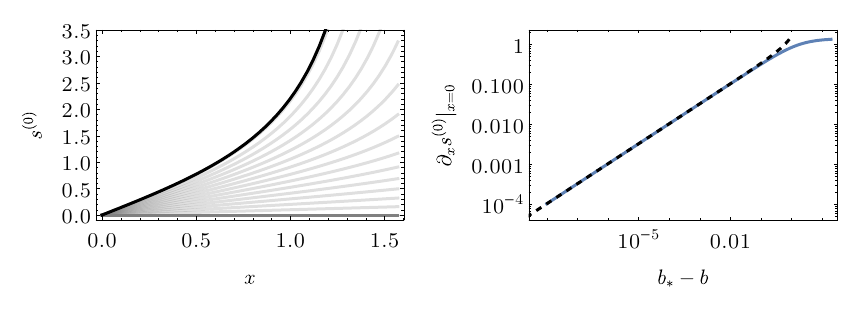}
	\caption{The static solution $\PhisZero$ for defocusing case exists only for $b<\bstar$. Left plot: the solution bifurcates from zero (dark grey line) for $b=\bstar$ and as $b$ decreases to minus infinity it converges to \eqref{eq:23.11.14_06} (black line). Right plot: the slope at the origin as a function of the Robin parameter $b$. Dashed line demonstrates agreement of the numerical results with the perturbative expansion, analogous to \eqref{eq:23.11.14_02}-\eqref{eq:23.11.14_05}, in the regime $\bstar-b\ll 1$.}
	\label{fig:RegularStaticConformalDefocusingPhiOrigin}
\end{figure}

We begin the discussion of the defocusing case with the following simple non-existence result. Let $\Phis$ be a nontrivial solution to \eqref{eq:08.11.23_1} with $\lambda=1$. We can multiply this equation by $\Phis$ and integrate the resulting expression over the domain $(0,\pi/2)$ getting
\begin{equation}
        \label{eq:08.11.23_2}
        \Phis(\pi/2)\,\Phis'(\pi/2)-\int_0^{\pi/2}\Phis'(x)^2\, \diff{x}  -\int_0^{\pi/2}\frac{\Phis^4(x)}{\sin^2 x}\, \diff{x}=0
        \,,
\end{equation}
where one of the boundary terms vanishes due to the regularity condition. If $\Phis$ satisfies Dirichlet or Neumann condition at $x=\pi/2$, then the first term in this expression vanishes. The remaining integral terms are all positive leading to the contradiction. The same reasoning also holds for the Robin boundary condition as long as $ \Phis(\pi/2)\,\Phis'(\pi/2) \leq 0$, which relates to $b$ being non-negative. The numerical indication shows that this nonexistence result can be extended to all $b>\bstar(\xh)$.

The construction and stability analysis of static solutions for the defocusing nonlinearity follow the steps used in previous section. Thus we restrict the discussion to the presentation of the main results.
Static solutions satisfying the Robin boundary condition \eqref{eq:231021_06} exist only for $b<\bstar=-2/\pi$. In such a case there exists only a node-less solution, which for consistency with notation used before we denote as $\PhisZero$. Its profile is a monotonically increasing function of $x$, see Fig.~\ref{fig:RegularStaticConformalDefocusingPhiOrigin}.
Analogously as for the $\lambda=-1$ case when $\bstar-b\ll 1$ the solution $\PhisZero$ can be constructed perturbatively and the result
is given by \eqref{eq:23.11.14_02}-\eqref{eq:23.11.14_05} with the minus sign in front of the $f_{3}$ term. In Fig.~\ref{fig:RegularStaticConformalDefocusingPhiOrigin} we plot the slope of $\PhisZero$ at the origin as a function of the Robin parameter and compare it with the perturbative solution.
As $b\rightarrow-\infty$ the solution approaches the exact singular solution
\begin{equation}
	\label{eq:23.11.14_06}
	\PhisZero_{-\infty}(x) = \sqrt{2}\tan{x}
	\,.
\end{equation}

\begin{figure}[t]
	\centering
	\includegraphics[width=0.5\textwidth]{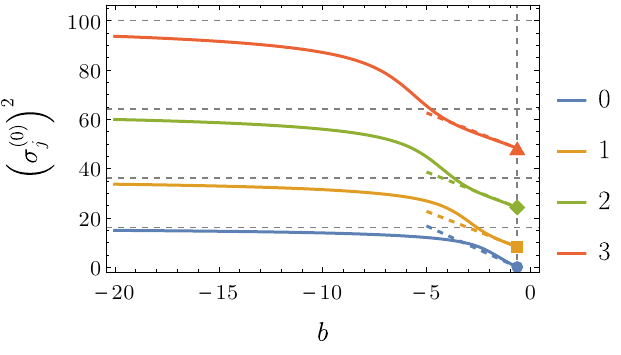}
	\caption{The dependence of the lowest eigenvalues (colour coded) of $\Phis^{(0)}$ on the Robin parameter $b$ for the defocusing case. The analogue of Fig.~\ref{fig:RegularStaticConformalFocusingEigenfrequencies}. In the limit $b\rightarrow -\infty$ we observe convergence towards eigenspectrum of the singular solution \eqref{eq:23.11.14_06}.}
	\label{fig:RegularStaticConformalDefocusingEigenfrequencies}
\end{figure}

Studying the linear perturbations of $\PhisZero$ we find only oscillatory modes ($\sigma^{2}>0$), see Fig.~\ref{fig:RegularStaticConformalDefocusingEigenfrequencies}, implying linear stability. Interestingly in the limit $b\rightarrow-\infty$ the spectrum approaches the sequence $(4+2j)^{2}$, $j=0,1,\ldots$. This could be understood when we consider linear problem around the singular solution \eqref{eq:23.11.14_06}. 
For the following analysis it is convenient to work with the original scalar field $\phi$ variable, cf. \eqref{eq:231021_05}. Note that $\phi=\sqrt{2}$, is an exact solution to \eqref{eq:231021_02} with $\lambda=-1$. A linear perturbation of this formal solution\footnote{We call it formal as it does not decay for $x\rightarrow\pi/2$; thus it becomes singular after rescaling the \eqref{eq:231021_02}.}
\begin{equation}
	\label{eq:23.11.19_01}
	\phi(t,x) = \sqrt{2} + e^{i\sigma t}\chi(x)
	\,,
\end{equation}
satisfies
\begin{equation}
	\label{eq:23.11.19_02}
	-\sigma^{2}\chi = \frac{1}{\tan^{2}{x}}\left(\tan^{2}{x}\,\chi'\right)' - \frac{4}{\cos^{2}{x}}\chi
	\,.
\end{equation}
A solution of \eqref{eq:23.11.19_02} which is smooth at $x=\pi/2$ is given in terms of the hypergeometric function \cite{NIST}
\begin{equation}
	\chi(x) = C\cos^{4}{x}\,\, _2F_1\left(2-\frac{\sigma }{2},\frac{\sigma }{2}+2;\frac{7}{2};\cos^{2}{x}\right)
   \,,
   \quad
   C\in\const
   \,,
\end{equation}
the other solution diverges as $(\pi/2-x)^{-1}$ as $x\rightarrow\pi/2$. Enforcing the regularity condition at the origin gives a condition for the eigenvalues
\begin{equation}
	\label{eq:23.11.19_03}
	\sigma^{2} = (4+2j)^{2}
	\,,
	\quad
	j=0,1,\ldots
	\,,
\end{equation}
which correspond to the limit of the spectrum of $\PhisZero$ when $b\rightarrow-\infty$.

\section{Static solutions in SAdS ($0<\xh<\infty$)}
\label{sec:SAdS}

\subsection{Linear equation ($\lambda=0$)}
\label{sec:LinearEquation}

We begin the study of static solutions in SAdS with the linear problem for equation \eqref{eq:22.09.22_09}. Let
\begin{equation}
	\Phi(v,y) = e^{i\omega v}\psi(y)
	\,.
\end{equation}
Under this separation of variables for $\lambda=0$ equation \eqref{eq:22.09.22_09} becomes
\begin{equation}\label{eq:27.10.23_01}
	2i\, \omega\, \partial_y \psi = L \psi \, ,
\end{equation}
where
\begin{equation}
	L\psi = \partial_{y}\left(y^{2}V\partial_{y}\psi\right) + y\,\partial_{y}V \psi + \frac{2}{y^{2}}\psi
	\,,
\end{equation}
while the Robin boundary condition \eqref{eq:22.09.22_11} is now given by
\begin{equation}\label{eq:27.10.23_03}
\psi'(0) - (b+i\omega) \psi(0)=0\, .
\end{equation}
Regular solutions of \eqref{eq:27.10.23_01} can be written explicitly using the Heun function\footnote{We use Wolfram Mathematica notation: $\HeunG(a,q,\alpha,\beta,\gamma,\delta,z)$ satisfies the general Heun equation: $z(z-1)(z-a)y''(z) + (\gamma  (z-1) (z-a)+\delta  z (z-a)+(z-1) z (\alpha +\beta
   -\gamma -\delta +1))y'(z) + (\alpha  \beta  z-q)y(z)=0$.} $\HeunG$ \cite{NIST}:
\begin{multline}
	\label{eq:27.10.23_02}
	\psi(x)=C\, \HeunG\left(\frac{\bar{\xi}}{\xi},\frac{2(1+\xh^2)}{\xi},1,1,1+\frac{2i\, \xh\, \omega}{3+\xh^2}, 1+\frac{\left(3+2\xh^2-i\sqrt{3+4\xh^2}\right)\xh\,\omega}{(3+\xh^2)\sqrt{3+4\xh^2}}, 
	\right.
	\\
	\left.
	\frac{2(\xh-y)(1+\xh^2)}{\xh\xi}\right)\, ,
\end{multline}
where $\xi = 3+2\xh^2+i\sqrt{3+4\xh^2}$ and $C$ is any constant. This solution can be plugged into \eqref{eq:27.10.23_03} to get a relation between $\xh$, $b$, and $\omega$. Then, for any fixed parameters $\xh$ and $b$ one can find numerically values of $\omega$ that can be interpreted as quasinormal frequencies of a zero solution. On the complex plane they form a set symmetric with respect to the imaginary axis but their locations, and as a consequence long term behaviour of perturbations of the zero solution, strongly depend on $\xh$ and $b$. For positive values of $b$ all quasinormal frequencies have non-zero real part and positive imaginary part, meaning that zero is a linearly stable solution. As $b$ decreases this situation changes, as can be seen in Fig.\ \ref{fig:Constant_xh}. At some point a pair of eigenvalues meets at the imaginary axis. The value of $b$ for which it takes place depends on $\xh$ and its dependence is shown in Fig.\ \ref{fig:PhasePlot} by the dotted line. As $b$ decreases further they separate again but remain at the imaginary axis. As these quasinormal frequencies move along the axis, eventually one of them crosses zero, meaning that there exists a solution with $\omega=0$, i.e., a zero mode. The value of $b$ for which this takes place will be denoted by $\bstar(\xh)$ and it is given by 
\begin{equation}\label{eq:25.10.23_01}
b_*(\xh)=-\frac{2(1+\xh^2)\,\HeunG'\!\left(\frac{\bar{\xi}}{\xi},\frac{2(1+\xh^2)}{\xi},1,1,1,1,\frac{2(1+\xh^2)}{\xi}\right)}{\xh \xi \,\HeunG\left(\frac{\bar{\xi}}{\xi},\frac{2(1+\xh^2)}{\xi},1,1,1,1,\frac{2(1+\xh^2)}{\xi}\right)}\, ,
\end{equation}
where $\HeunG'$ is a derivative of the Heun function with respect of its last variable. For smaller values of $b$, one of the eigenvalues has negative imaginary part, implying instability, cf. \cite{Holzegel.2014pzh}. This change of behaviour takes place at the critical curve $b_*(\xh)$ plotted in Fig.\ \ref{fig:PhasePlot} as a solid line. Values of $b_*(\xh)$ are strictly negative, it attains a maximum value of $-0.603341$ at $\xh\approx 2.08026$, for $\xh\rightarrow 0$ it diverges to $-\infty$ as $\bstar(\xh)\sim \xh^{-1}$, while for $\xh\rightarrow \infty$ it converges to $-2/\pi$. Since this curve separates regions of the phase space where zero solution is stable and unstable, it lets us conclude that for large positive values of $b$ zero is stable, while for large negative it is generically unstable. Additionally, large black holes tend to posses stable zero solution since for any $b$ taking $\xh$ small enough lands us in the basin of stability.

\begin{figure}[t]
	\centering
	\includegraphics[width=0.5\textwidth]{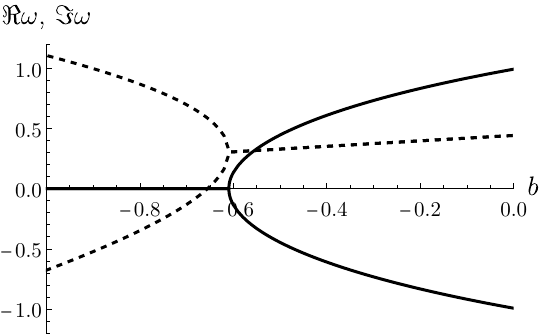}
	\caption{Plot of the behaviour of two lowest quasinormal frequencies for $\xh=1$ as $b$ varies. The solid and dashed lines present real and imaginary parts, respectively.}
	\label{fig:Constant_xh}
\end{figure}

\begin{figure}[t]
	\centering
	\includegraphics[width=0.5\textwidth]{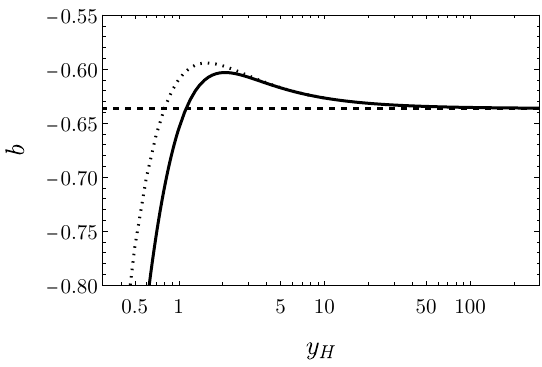}
	\caption{Plot of the phase space $(\xh, b)$. The solid line denotes the critical curve while the dotted one is the bifurcation curve. The horizontal dashed line shows the asymptotic convergence to $-2/\pi$  as  $\xh\rightarrow \infty$.}
	\label{fig:PhasePlot}
\end{figure}

The behaviour of $b_*(\xh)$ as $\xh\to 0$ is a straightforward implication of \eqref{eq:25.10.23_01} so let us now focus on the case $\xh\to\infty$. Here, the explicit formula for  $b_*(\xh)$ seems to be of little help since the asymptotic behaviour of $\HeunG$ and its derivative is, to our best knowledge, not well known in this limit. Instead, we can use the fact that \eqref{eq:231021_02} for static solutions comes not only as a result of assuming an AdS background but can also be obtained as $\xh\to\infty$ limit of equation $L\psi=0$. To see it, it is convenient to work with the original radial coordinate $r$ defined in $[r_H,\infty)$. Let us compactify this interval by introducing $x$ such that $\tan x = r-r_H$ with $x\in[0,\pi/2)$. In this new coordinate static solutions of the nonlinear problem must satisfy
\begin{equation}\label{eq:25.09.23_5a}
\frac{1}{\tilde{\mu}(x)} \frac{\diff{}}{\diff{x}}\left(\tilde{\mu}(x) \frac{\diff{\phi}}{\diff{x}}\right)+\tilde{\nu}(x)\left(2\phi-\lambda \phi^3\right)=0
	\,,
\end{equation}
where
\begin{align}\label{eq:25.09.23_6}
\tilde{\mu}(x)&=\left[1+\frac{3}{2}r_H \sin 2x +3r_H^2\cos^2 x \right](r_H+\tan x)\tan x,\\
 \tilde{\nu}(x)&=\frac{r_H+\tan x}{\sin x \cos x\left[1+\frac{3}{2}r_H \sin 2x +3r_H^2\cos^2 x \right]}
 \,.
\end{align}
This equation is regular under the limit $r_H\to 0$ and leads to
\begin{equation}\label{eq:25.09.23_5b}
\frac{1}{\tan^2 x} \frac{\diff{}}{\diff{x}}\left(\tan^2 x \frac{\diff{\phi}}{\diff{x}}\right)+\frac{1}{\cos^2 x}\left(2\phi-\lambda \phi^3\right)=0
	\,.
\end{equation}
For $\lambda=0$ we get the desired result, however, this argument gives us also a correspondence between static solutions of nonlinear equation \eqref{eq:22.09.22_09} with large $\xh$ and equation \eqref{eq:231021_02}.

Our description above regards change of behaviour as $\xh$ is fixed and one varies $b$. In the opposite situation, when $b$ is fixed and $\xh$ changes, the behaviour is more complicated and strongly depends on $b$, as can be seen in Fig.\ \ref{fig:Constant_b}. When the line of constant $b$ lays entirely above the bifurcation curve in Fig.\ \ref{fig:PhasePlot}, the imaginary part of the lowest quasinormal frequency is positive and decreases monotonically as $\xh$ increases. If this line crosses the bifurcation curve, the imaginary part of the frequency bifurcates at some $\xh$. When the line additionally crosses the critical curve, one of the emerging branches becomes at some point negative. The further behaviour of the imaginary part of the lowest frequency depends on whether $b>-2/\pi$ or $b<-2/\pi$. In the first case, two branches eventually merge and converge to some positive values as $\xh\to\infty$. Otherwise, they stay separated and one of them remains negative.

\begin{figure}[t]
	\centering
	\includegraphics[width=1\textwidth]{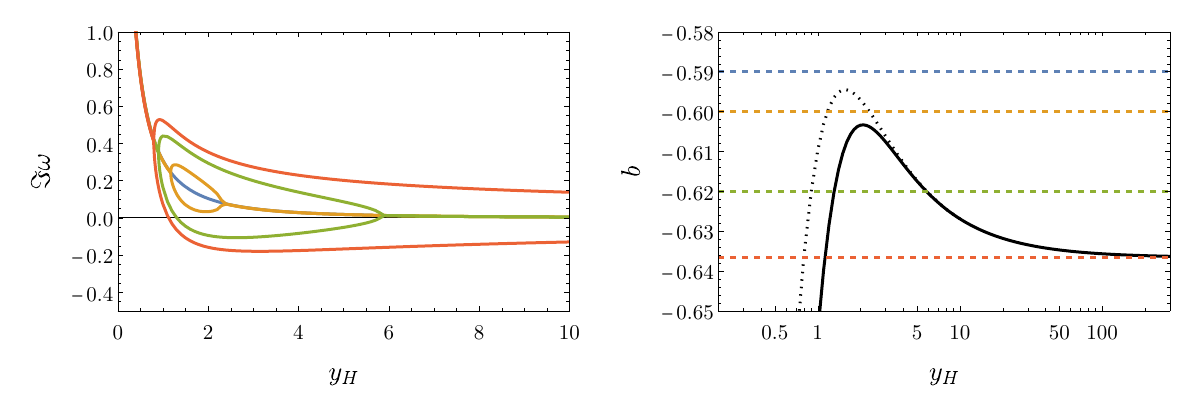}
	\caption{Plot of the behaviour of the imaginary part of lowest quasinormal frequencies for various $b$ ($b=-0.59$, $-0.6$, $-0.62$, and $-\pi/2$) as $\xh$ varies.}
	\label{fig:Constant_b}
\end{figure}

\subsection{Focusing case ($\lambda=-1$)}
\label{sec:FocusingCase}

\subsubsection{Existence}
\label{sec:ExistenceStaticSolutionsFocusingCase}

Since the shooting method is one of the main tools that we use in the investigations of static solutions to \eqref{eq:22.09.22_09} we begin by showing that it is well posed.  By putting $\Phi(v,y)=\Phis(y)$ in \eqref{eq:22.09.22_09} we get
\begin{equation}
        \label{eq:06.09.23_1}
        \frac{\diff{}}{\diff{y}}\left(\mu(y)\, \frac{\diff{\Phis(y)}}{\diff{y}}\right)-y\left(1+\xh^2\right)\Phis(y)-\lambda\, \xh^3 \Phis(y)^3=0
        \,,
\end{equation}
where
\begin{equation}
        \label{eq:06.09.23_2}
        \mu(y)=\xh^3+x^2 \xh^3-y^3\left(1+\xh^2\right)
        \,,
\end{equation}
is a non-negative function. Since $\mu$ behaves near $y=\xh$ like $\mathcal{O}(\xh-y)$, this equation is singular there. Hence, one need to impose the appropriate regularity conditions: if $\Phis(\xh)=c$, then 
\begin{equation}
        \label{eq:13.09.23_2}
        \Phis'(\xh)=-\frac{1+\xh^2 \left( 1+\lambda\, c^2\right)}{\xh \left(3+\xh^2\right)}\, c
        \,.
\end{equation}
This condition ensures local existence of the solution.
\begin{lem}\label{thm:25.09.23_1}
Equation \eqref{eq:06.09.23_1} with $\Phis(\xh)=c$ and $\Phis'(\xh)$ given by \eqref{eq:13.09.23_2} constitutes a well defined initial value problem: for every $c\in \mathbb{R}$ it has a unique local solution on some interval $(\xh-\varepsilon,\xh)$ and this solution depends continuously on $c$ and $\xh$.
\end{lem}
\begin{proof}
Let us introduce new variables:
\begin{equation}
        \label{eq:13.09.23_1a}
        t:=\sqrt{1-\frac{y}{\xh}}
        \,,
        \qquad
        w(t):=\frac{\Phis(y)}{\sqrt{p(t)}}
        \,,
\end{equation}
where
\begin{equation}
        \label{eq:13.09.23_4}
p(t)=3+\xh^2- t^2\left(3+2\xh^2\right) +t^4\left(1+\xh^2\right)\,,
\end{equation}
is strictly positive for $t\in[0,1]$. Now \eqref{eq:06.09.23_1} can be written as
\begin{equation}
        \label{eq:13.09.23_1b}
        \frac{1}{t} \frac{\diff{}}{\diff{t}} \left(t \, \frac{\diff{w}}{\diff{t}} \right)- \frac{1}{p(t)^2}\left[ q(t) w+4\lambda\, \xh^2 w^3\right]=0
        \,,
\end{equation}
where
\begin{equation}
        \label{eq:13.09.23_3}
q(t)=2\left(3+\xh^2\right)-\left(9+8\xh^2\right)t^2+2\left(1+\xh^2\right)t^4
\,,
\end{equation}
is a polynomial with $q(0)=2\left(3+\xh^2\right)>0$, $q(1)=-\left(1+4\xh^2\right)<0$ and a single zero in between. In these variables the differential operator present in our equation becomes a radially-symmetric two-dimensional laplacian with $t=0$ being the singular point. The regularity condition now transforms to simple $w'(0)=0$ while the relation between $\Phis$ and $w$ leads to $w(0)=\sqrt{3+\xh^2}\, c$ for $\Phis(\xh)=c$. The lemma follows from standard results regarding this type of problems \cite{coddington1955theory, ni1985uniqueness, hastings2011classical}.

\end{proof}

In the focusing case the global existence is also ensured. It implies that to every value of $c$ we can  assign some $b$. It can be either finite or infinite, in the case of solution satisfying Dirichlet BC at $y=0$.
\begin{lem}
For $\lambda=-1$ solutions given by Lemma \ref{thm:25.09.23_1} can be extended to the whole interval $(0,\xh)$.
\end{lem}
\begin{proof}
By Lemma \ref{thm:25.09.23_1} we know that $w$ exists in some interval $(0,2\varepsilon)$ so let us focus here on the interval $I=(\varepsilon,1)$. Equation \eqref{eq:13.09.23_1b} can be reformulated by introducing new variable
\begin{equation}
        \label{eq:05.10.23_1}
u(t)=\frac{t^{1/4}}{\sqrt{t^4+3\left(1-t^2\right)+\left(1-t^2\right)^2 \xh^2}}w(t)\,,
\end{equation}
so it becomes
\begin{equation}
        \label{eq:05.10.23_2}
\frac{\diff{}}{\diff{t}}\left(a(t)\, \frac{\diff{u}}{\diff{t}}\right) + r(t)u+4\xh^2 u^3=0\,,
\end{equation}
where
\begin{align}
        \label{eq:05.10.23_3}
a(t)=\sqrt{t} \left(t^4+3\left(1-t^2\right)+\left(1-t^2\right)^2 \xh^2\right)\,,
\end{align}
is positive and $r$ is continuous in $I$. For any solution $v(t)$ we can define the following functional
\begin{align}
        \label{eq:05.10.23_4}
V(t)=a(t)u'^2+2u^2+2\xh^2 u^4\,.
\end{align}
It is non-negative and its derivative is
\begin{align}
        \label{eq:05.10.23_5}
V'(t)&=a'(t)u'^2+2a(t)u''u'+4uu'+8\xh^2 u^3u'=-a'(t)u'^2+\left[4-2r(t)\right]u u'\,,
\end{align}
where we have used the equation of motion \eqref{eq:05.10.23_2}. Let us fix $M$ and $m$ such that $|a'(t)|<M$ and $0<m<a(t)$ for $t\in I$, then we can bound $V'(t)$ as follows:
\begin{align*}
V'(t)&\leq M\, u'^2+\left[2-r(t)\right]\left(u^2+ u'^2\right)\leq \left(M+2+|r(t)|\right)u'^2 +\left(2+|r(t)|\right)u^2\\
&\leq \frac{M+2+|r(t)|}{m} \, a(t) u'^2 +\left(1+\frac{|r(t)|}{2}\right)2u^2 \\
&\leq  \left(1+\frac{|r(t)|}{2}+\frac{M+2+|r(t)|}{m}\right)\left( a(t) u'^2+2u^2\right)\\
&\leq\left(1+\frac{|r(t)|}{2}+\frac{M+2+|r(t)|}{m}\right)V(t)\,.
\end{align*}
Since $r(t)$ is bounded in $I$, it gives us an upper bound on the rate of increase of $V$ implying existence of $V$ bounded by some $K$ in the whole interval. Since 
\begin{align}
        \label{eq:05.10.23_7}
m u'^2 + 2 u ^2 \leq a(t) u'^2 + 2 u^2 \leq V(t) \leq K\,,
\end{align}
we get
\begin{align}
        \label{eq:05.10.23_8}
u \leq \sqrt{\frac{K}{2}}\,, \qquad u'\leq \sqrt{\frac{K}{m}}\,.
\end{align}
By going back to the original variable $\Phis$, together with Lemma \ref{thm:25.09.23_1} we conclude existence of the solution in the interval $(0,\xh)$.
\end{proof}

\subsubsection{Construction}
\label{sec:ConstructionStaticSolutionsFocusingCase}

\begin{figure}[t]
	\centering
	\includegraphics[width=1\textwidth]{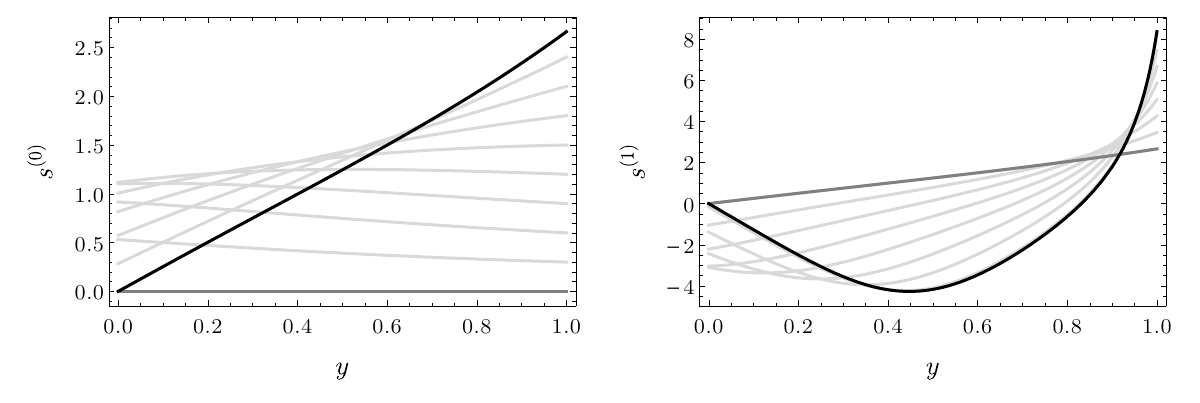}
	\caption{Profiles of static solutions $\Phis^{(0)}(y)$ and $\Phis^{(1)}(y)$ with horizon radius $\xh=1$ and focusing nonlinearity for different values of the Robin parameter, cf. Fig.~\ref{fig:StaticConformalFocusingGlobal} for the horizonless case. Left: the solution $s^{(0)}$ bifurcates from zero for $b=\bstar(\xh)$ and tends to the solution satisfying Dirichlet BC (black line) as $b\rightarrow\infty$. Note the node-less solution does not exist in the region $b<\bstar(\xh)$ (see the text). Right: the first excited state interpolates between two static solutions satisfying the Dirichlet BC (dark grey and black lines) as $b$ goes from minus to plus infinity.
	}
	\label{fig:StaticConformalFocusingPhiProfiles}
\end{figure}

We construct static solutions numerically using the shooting technique which proceeds as follow. 
For $\Phi(v,y)=\Phis(y)$ the equation \eqref{eq:22.09.22_09} (with $\lambda=-1$) becomes:
\begin{equation}
	\label{eq:23.10.25_01}
	- \partial_{y}\left(y^{2}V\partial_{y}\Phis\right) - \left(y\,\partial_{y}V + \frac{1}{2y^{2}}\right)\Phis - \Phis^{3} = 0
	\,.
\end{equation}
A regular local solution at the horizon $y=\xh$ satisfies
\begin{equation}
	\label{eq:23.10.25_02}
	\Phis(y) = f + \frac{f \left(\left(f^2-1\right) \xh^2-1\right)}{\xh
   \left(\xh^2+3\right)}(y-\xh) + \mathcal{O}\left((y-\xh)^{2}\right)
   \,,
\end{equation}
where $f$ is a free parameter which uniquely determines the solution. As proved above local solutions extend smoothly to $y\in[0,\xh]$. Choosing $f$ we integrate \eqref{eq:23.10.25_01} starting at $y=\xh$ toward $y=0$ with data which follows from \eqref{eq:23.10.25_02}. As $y=\xh$ is a regular point of the equation, for any $f$ we get a solution satisfying the Robin BC $\left(\partial_{y}\Phis-b\Phis\right)|_{y=0}=0$ with some $b=b(f)$. For small $f$ the solution is monotonic, but as we increase $f$ solution becomes oscillatory and it oscillates finite number of times in the interval $y\in[0,\xh]$.

By adjusting $f$ we can construct solutions with prescribed $b$. It turns out that using this procedure we find that, for a given value of the Robin parameter $b$, there exist infinitely many solutions $\Phis^{(n)}$, where $n$ is non-negative integer (nodal index) which enumerates the number of nodes in the solution. However, the node-less solution $n=0$ exists only for $b>\bstar(\xh)$.

Small node-less solutions can be also constructed perturbatively, analogously to the regular case, see Sec.~\ref{sec:RegularCaseFocusingCase}. Since linear problem with finite $\xh$ has solutions expressed by the Heun function, this time we are not able to write explicitly formula for $f_3$ in expansion \eqref{eq:23.11.14_02}. However, this function and solutions in higher orders can be found numerically.

Fig.~\ref{fig:StaticConformalFocusingPhiProfiles} illustrates the lowest static solutions, $\Phis^{(0)}$ and $\Phis^{(1)}$, for $\xh$ fixed when $b$ changes.  
Increasing $b$ solutions approach the zero value at the conformal boundary ($y=0$) and in the limit $b\rightarrow\infty$ they converge to the respective solutions satisfying the Dirichlet BC.
Alternatively, keeping $b$ fixed and increasing $\xh$ (small BH), the solutions start to resemble the respective solutions of the horizon-less case, and in fact approach them in the limit $\xh\rightarrow\infty$ (cf. the regular case in Sec.~\ref{sec:RegularCase}). Whereas for $\xh\rightarrow 0$ (large BH) the solutions grow in magnitude as $1/\xh$ such that $\xh\Phis$ approaches the limiting solution on the appropriately scaled interval. This limit is discussed in Sec.~\ref{eq:StaticSolutionsLargeBH}.

\subsubsection{Linear stability}
\label{sec:LinearStabilityStaticSolutionsFocusingCase}

Next, to determine the role of static solutions $\Phis^{(n)}(y)$ in dynamics we study their linear stability. 
Therefore we write
\begin{equation}
	\label{eq:23.10.26_01}
	\Phi(v,y) = \Phis(y) + e^{i\sigma v}\chi(y)
	\,.
\end{equation}
Plugging this into the equation \eqref{eq:22.09.22_09} and neglecting nonlinear terms in $\chi$ we obtain the eigenvalue problem
\begin{equation}
	\label{eq:23.10.26_02}
	2i\sigma\chi = \partial_{y}\left(y^{2}V\partial_{y}\chi\right) + \left(y\,\partial_{y}V - \frac{1}{2y^{2}}\right)\chi + 3\Phis^{2}\chi = 0
	\,,
\end{equation}
subject to the boundary condition 
\begin{equation}
	\label{eq:23.10.26_03}
	\left.\left(- i\sigma \chi + \partial_{y}\chi - b \chi\right)\right|_{y=0} = 0
	\,,
\end{equation}
which follows from \eqref{eq:22.09.22_11}.

Local solutions of \eqref{eq:23.10.26_02} at $y=\xh$ behave as
\begin{equation}
	\label{eq:23.10.26_04}
	\chi(y) = h+h\frac{\left(3 f^2-1\right) \xh-\xh^{-1}}{\xh^2+2 i \sigma \xh+3} (y-\xh)+\mathcal{O}\left((y-\xh)^2\right)
	\,,
	\quad
	h=\chi(\xh)
	\,,
\end{equation}
with $h=\chi(\xh)$ an arbitrary complex constant, and since at $y=0$ the equation is regular the solution is smooth there. To find a solution of \eqref{eq:23.10.26_02} satisfying \eqref{eq:23.10.26_03} one could again use the shooting technique, this time on a complex plane. Given a guess for $\sigma$ one can start the integration from $y=\xh$ with regularity conditions \eqref{eq:23.10.26_04}. It is clear that the condition \eqref{eq:23.10.26_03} with fixed value of $b$ can be satisfied only for a discrete set of $\sigma$. Such solution, a QNM, by definition is an outgoing solution at the black hole horizon. Although this procedure can be made semi-automatic, we have found that it is increasingly difficult to accurately determine modes with large $|\Im\sigma|$ (both higher overtones and unstable modes). Thus, as an alternative approach to find QNMs we use the method of \cite{Bizoń.2020sul}, which after a suitable discretisation turns \eqref{eq:23.10.26_02} and \eqref{eq:23.10.26_03} into an algebraic eigenvalue problem. Another advantage of this approach is that it gives us the whole spectrum of QNMs without the need to provide a guess as for the shooting method. We have checked that the two approaches provide consistent results.

\begin{figure}[t]
	\centering
	\includegraphics[width=0.47\textwidth]{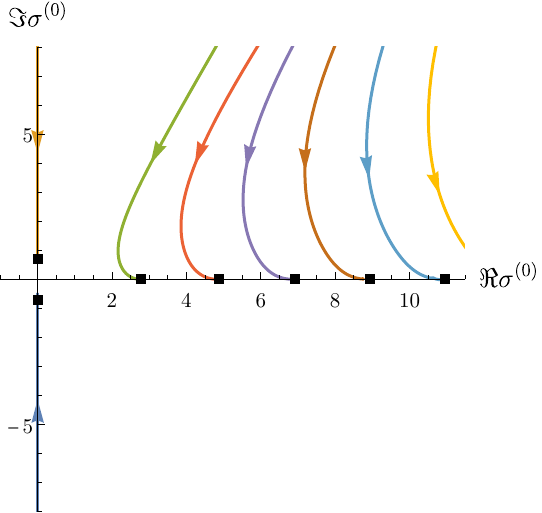}
	\hspace{4ex}
	\includegraphics[width=0.47\textwidth]{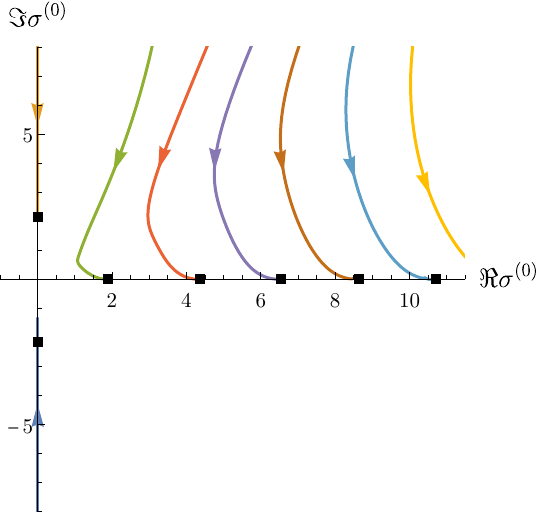}
	\caption{The lowest QNMs of $\Phis^{(0)}$ (focusing case) as a function of $\xh$ for fixed $b$: $-1/2$ (left plot) and $1$ (right plot). When $\xh$ increases (indicated by arrows) oscillatory modes approach the real axis, and in the limit $\xh\rightarrow\infty$ they converge to the corresponding QNMs of the horizonless solution (squares). The modes on the imaginary axis also approach the corresponding modes of the horizonless solution, though this approach is not monotonic, see Fig.~\ref{fig:StaticConformalFocusingQNMUnstable} for details for the unstable mode $\sigma^{(0)}_{-1}$.}
	\label{fig:StaticConformalFocusingQNM}
\end{figure}

\begin{figure}[t]
	\centering
	\includegraphics[width=0.47\textwidth]{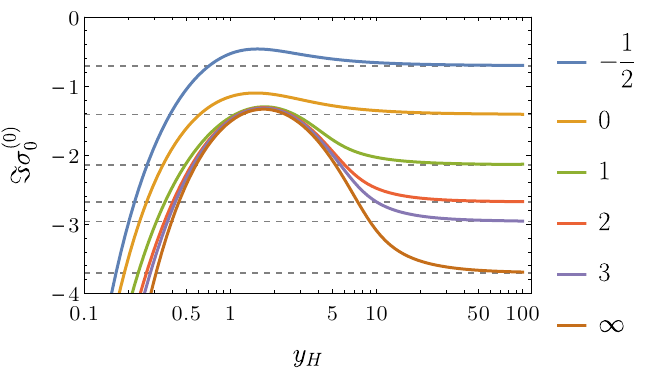}
	\hspace{4ex}
	\includegraphics[width=0.47\textwidth]{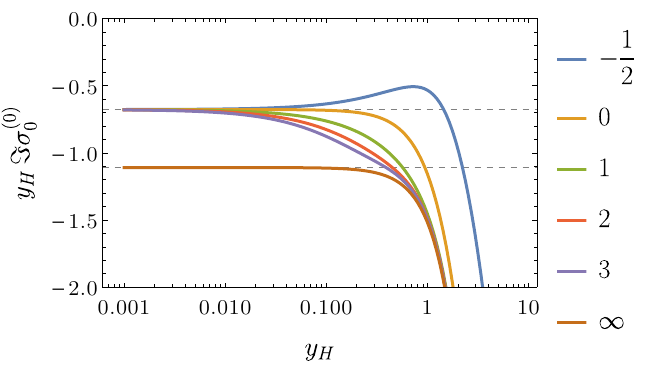}
	\caption{The unstable mode $\sigma^{(0)}_{0}$ as a function of $\xh$ for fixed $b=-1/2, 0, 1, 2, 3$ and $\infty$ (Dirichlet BC). Left plot: as $\xh\rightarrow\infty$ we see convergence to the unstable mode of the corresponding horizonless solution (horizontal dashed lines). Right plot: for $\xh\rightarrow 0$ the product $\xh \sigma^{(0)}_{0}$ approaches the eigenmodes of the solution $S^{(0)}(z)$ of \eqref{eq:23.10.25_05}, with either Robin or Dirichlet BC for $b<\infty$ and $b=\infty$ respectively, see discussion in Sec.~\ref{eq:StaticSolutionsLargeBH}.}
	\label{fig:StaticConformalFocusingQNMUnstable}
\end{figure}

\begin{figure}[t]
	\centering
	\includegraphics[width=1.0\textwidth]{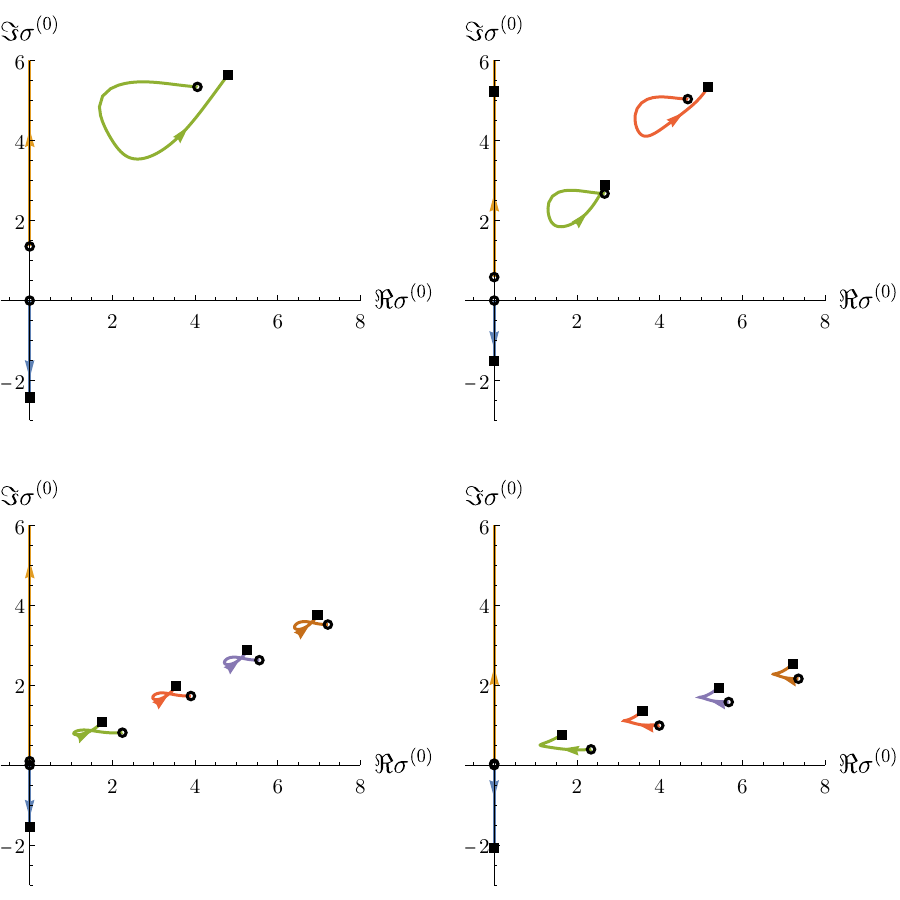}
	\caption{The lowest QNMs of $\Phis^{(0)}$ (focusing case) for fixed $\xh=1/2$, $1$, $3$, and $5$ (from left to right and from top to bottom) as a function of $b$. As $b$ increases from $\bstar(\xh)$ to infinity (indicated by arrows) we observe convergence towards the QNMs spectrum of the problem with Dirichlet BC (squares). Whereas when $b$ approaches $\bstar(\xh)$ from above the spectrum $\sigma^{(0)}_{j}$ converges to the eigenfrequencies $\omega_{j}$ of the trivial solution $\Phi=0$ (circles).}
	\label{fig:StaticConformalFocusingQNM2}
\end{figure}

The linear stability analysis strongly suggests that solution $\Phis^{(n)}$ has exactly $n+1$ unstable modes\footnote{We use the convention consistent with \eqref{eq:23.11.18_20}.}
\begin{equation}
	\label{eq:23.10.30_01}
	\Im\sigma^{(n)}_{-n}<\Im\sigma^{(n)}_{-(n-1)},\ldots,<\Im\sigma^{(n)}_{0} < 0
	\,.
\end{equation}
This implies that all nontrivial static solutions are linearly unstable. Beside the negative modes \eqref{eq:23.10.30_01} the solution  $\Phis^{(n)}$ has $n+1$ purely damped modes $\Im\sigma>0$, $\Re\sigma=0$ and an infinite number of (stable) oscillatory modes $\Im\sigma>0$, $\Re\sigma\neq 0$.

Convergence of the solution profiles when $\xh\rightarrow\infty$ and $b\rightarrow \infty$ (independently) to the horizonless solutions and solutions satisfying Dirichlet BC respectively, is transferred to the convergence of the QNMs spectrum. Also, the behaviour of the spectrum when $\xh\rightarrow 0$ can be understood by considering the linear perturbations of the limiting solutions $S^{(n)}(z)$, see Sec.~\ref{eq:StaticSolutionsLargeBH}. Anticipating the results of the nonlinear evolution we discuss in some detail the dependence of the QNMs spectrum of the node-less solution $\PhisZero$ on the parameters $(\xh,b)$.

First consider fixed $b$ above the critical line, where $\Phis^{(0)}$ does exist. As $\xh\rightarrow\infty$ we observe convergence of the spectrum $\sigma^{(0)}$ to the modes of the the horizonless solution. When $\xh\rightarrow 0$ the modes grow in magnitude as $\xh^{-1}$. Close inspection of the rescaled spectrum, $\xh\sigma^{(n)}$, reveals that it converges in the limit $\xh\rightarrow 0$ to the linear spectrum of the solutions \eqref{eq:23.10.25_05}, and that this limit is universal for Robin BC with $b<\infty$ (but it differs from the Dirichlet BC, i.e., the $b=\infty$ case).
Next, taking the limit $b\rightarrow\infty$ with $\xh$ fixed the spectrum converges to the respective spectrum with the Dirichlet BC. On the other extreme when we approach the critical curve $\bstar(\xh)$ from above we observe that the spectrum of $\Phis^{(0)}$ approaches the spectrum of $\Phi=0$, from which $\Phis^{(0)}$ bifurcates.
This behaviour of the QNMs of $\Phis^{(0)}$ is illustrated in Fig.~\ref{fig:StaticConformalFocusingQNM}-\ref{fig:StaticConformalFocusingQNM2}.

In the regime where $\Phi=0$ is linearly stable $b>\bstar(\xh)$ it is expected that for small initial data this solution will act as an attractor, whereas large data will blowup. In Sec.~\ref{sec:Dynamics} we provide a numerical evidence that $\Phis^{(0)}$ separates these two scenarios. Interestingly, in the $b<\bstar(\xh)$ case none of the static solutions is linearly stable, including the trivial solution $\Phi=0$. Therefore, it is natural to expect that arbitrarily small generic initial data will end up blowing up. Details of this unstable dynamics are presented in Sec.~\ref{sec:SmallbDynamicsFocusingCase}.

\subsection{Defocusing case ($\lambda=1$)}
\label{sec:StaticDefocusingCase}
Using the same method as in Sec.~\ref{sec:RegularCase}, one can show that for defocusing nonlinearity there are no static solutions with $b\geq 0$. Multiplying \eqref{eq:06.09.23_1} by $s$ and integrating it over the interval $[0,x_H]$ here leads to
\begin{equation}
        \label{eq:06.09.23_3}
        -s(0)s'(0)x_H^3-\int_0^{x_H}\mu(x)s'(x)^2\, \diff{x} -\int_0^{x_H}x(1+x_H^2)s(x)^2\, \diff{x} - x_H^3\int_0^{x_H}s(x)^4\, \diff{x}   =0\,,
\end{equation}
where we have used the fact that $\mu(x_H)=0$ and $\mu(0)=x_H^3$. As before, this identity cannot hold for $b\geq 0$ leading to the nonexistence. Similarly as in the regular case, it is only a partial result: the nonexistence of static solutions seems to hold for all $b>\bstar(\xh)$.

\begin{figure}[t]
	\centering
	\includegraphics[width=0.47\textwidth]{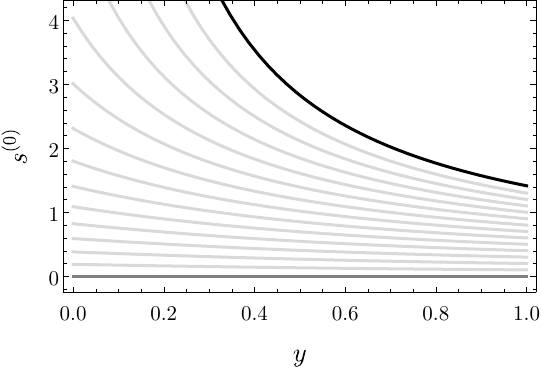}
	\hspace{4ex}
	\includegraphics[width=0.47\textwidth]{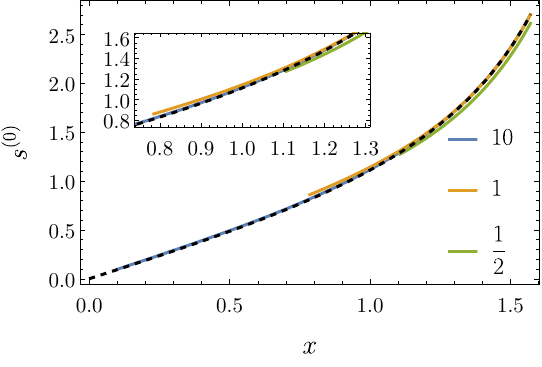}
	\caption{The static solution $\PhisZero$ for defocusing case exists only for $b<\bstar(\xh)$. Left plot: solution bifurcates from zero (dark grey line) for $b = \bstar(\xh)$ and as $b$ decreases to minus infinity it converges to $\sqrt{2}/y$ (black line) independently of $\xh$ (here $\xh=1$) cf. Fig.~\ref{fig:RegularStaticConformalDefocusingPhiOrigin} for regular case. Right plot: for fixed $b$ and $\xh$ going to infinity the solution approaches the corresponding solution of the regular problem (dashed line). Here we plot the results for $b=-2$ and $\xh=1/2,1,10$ (colour coded) with respect the global coordinate $x$ introduced in Sec.~\ref{sec:RegularCase}. Note that for finite $\xh$ the range of $x$ is $[\arctan(1/\xh),\pi/2]$.}
	\label{fig:StaticConformalDefocusingPhi}
\end{figure}

For each $b<\bstar(\xh)$ there exists a single static solution $\Phis^{(0)}(y)$, which is a monotonically increasing node-less function. The profile of that solution, and in particular its dependence on $\xh$ and $b$ is presented in Fig.~\ref{fig:StaticConformalDefocusingPhi}. In the two extreme cases $\xh$ fixed and $b\rightarrow -\infty$ and $b$ fixed and $\xh\rightarrow\infty$ this solution converges respectively to the singular solution $\sqrt{2}/y$ (irrespectively of $\xh$) and the regular solution $s^{(0)}(x)$ studied in detail in Sec.~\ref{sec:RegularCaseDefocusingCase}. Analogously to the focusing case, close to the critical curve $b=\bstar(\xh)-\varepsilon$ ($0<\varepsilon\ll 1$), one can compute the solution $\Phis^{(0)}$ and study its linear stability perturbatively in $\varepsilon$. For the same reasons as before, we skip presenting this calculation.

The linear stability analysis of $\Phis^{(0)}(y)$ is analogous to the procedure for the focusing case presented in Sec.~\ref{sec:LinearStabilityStaticSolutionsFocusingCase}. We find that the spectrum $\sigma^{(0)}$ consists of stable modes only, i.e. $\Im\sigma^{(0)}>0$, which implies linear stability of $\Phis^{(0)}(y)$. Recall that in the region $b<\bstar(\xh)$ the zero solution is linearly unstable $\Im\omega<0$. Thus, we have a classical pitchfork bifurcation at $b=\bstar(\xh)$, an analogue of the self-gravitating case \cite{Bizoń.2020}. 
Although the structure of QNMs depends in a nontrivial way on both $\xh$ and $b$, e.g. see Figs.~\ref{fig:StaticConformalDefocusingQNM} and \ref{fig:StaticConformalDefocusingQNM2}, the behaviour for fixed $\xh$ and $b\rightarrow-\infty$ and also for fixed $b$ with $\xh\rightarrow\infty$ can be readily understood based on the behaviour of the static solution itself.

For $(\xh,b)$ close to the critical curve $\bstar(\xh)$ the QN frequencies $\sigma^{(0)}$ are close to the spectrum of the trivial solution $\omega$.
In particular, one of the two modes on the imaginary axis starts at the origin for $b=\bstar(\xh)$. 
For fixed $b<-2/\pi$ and $\xh$ increasing from  $\xh^{-1}(\bstar=b)$ to infinity we see smooth transition form $\omega$ to the spectrum of the regular solution with the same Robin parameter. For $-2/\pi<b<\bstar(\xh)$ we obtain $\omega$ (different for different $\xh$) on both ends of the interval allowed for $\xh$.
Alternatively, with $\xh$ fixed the QNMs move smoothly from $\omega$ for $b=\bstar(\xh)$ towards the spectrum of the singular solution $\sqrt{2}/y$ as $b$ goes to $-\infty$, c.f. Fig.~\ref{fig:StaticConformalDefocusingQNM2}. The linear perturbation of this solution is described by the equation
\begin{multline}
        \label{eq:16.11.23_1}
       2i \xh^3 \omega\chi' = (\xh-y)\left[(1+\xh^2)y^2+\xh y+\xh^2\right]\chi''
        \\
        +\left[2\xh^3 y-3(1+\xh^2)y^2\right]\chi'-\left[(1+\xh^2)y+\frac{6\xh^2}{y^2}\right]\chi\,.
\end{multline}
The spectrum $\omega$ can be then found numerically with the help of the Leaver's-type method, similarly as for \ \eqref{eq:23.11.07_01}, by expanding $\chi$ into a series around $y=\xh$ and demanding that the coefficients of the expansion converge to zero. However the rate of convergence of this method degrades when $\xh$ increases. This follows form the fact that the singularities in the equation approach $\pm i$, which become close the disc of convergence $|w-\xh|\leq \xh$, $w\in\mathbb{C}$. Therefore for large $\xh$ we employ the algebraic method, see \cite{Bizoń.2020sul}. 

\begin{figure}[t]
	\centering
	\includegraphics[width=0.47\textwidth]{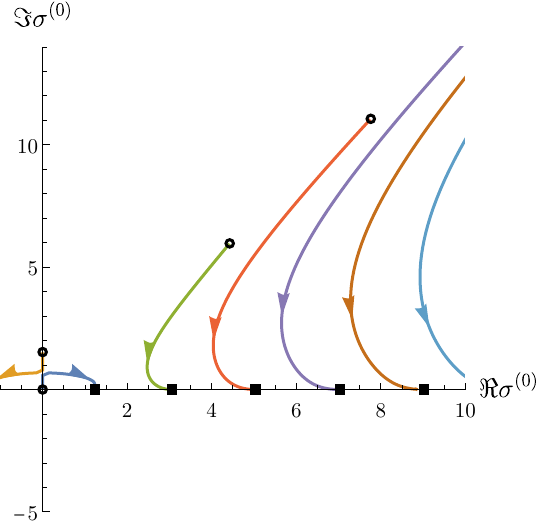}
	\hspace{4ex}
	\includegraphics[width=0.47\textwidth]{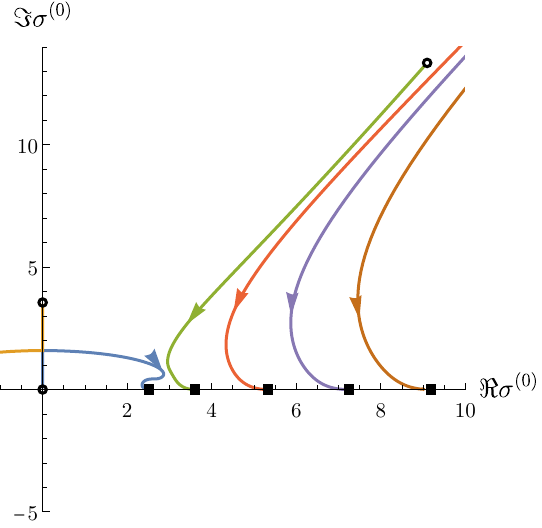}
	\caption{The lowest QNMs of $\Phis^{(0)}$ (defocusing case) for fixed $b=-1$ (left plot) and $b=-2$ (right plot) as a function of $\xh\in[\xh^{-1}(\bstar=b),\infty)$. Modes start at the critical line from the modes of the trivial solution (circles). As $\xh$ increases (indicated by arrows) the QNMs approach the real axis, and in the limit $\xh\rightarrow\infty$ they converge to the corresponding QNMs of the horizonless solution (squares).}
	\label{fig:StaticConformalDefocusingQNM}
\end{figure}

\begin{figure}[t]
	\centering
	\includegraphics[width=1.0\textwidth]{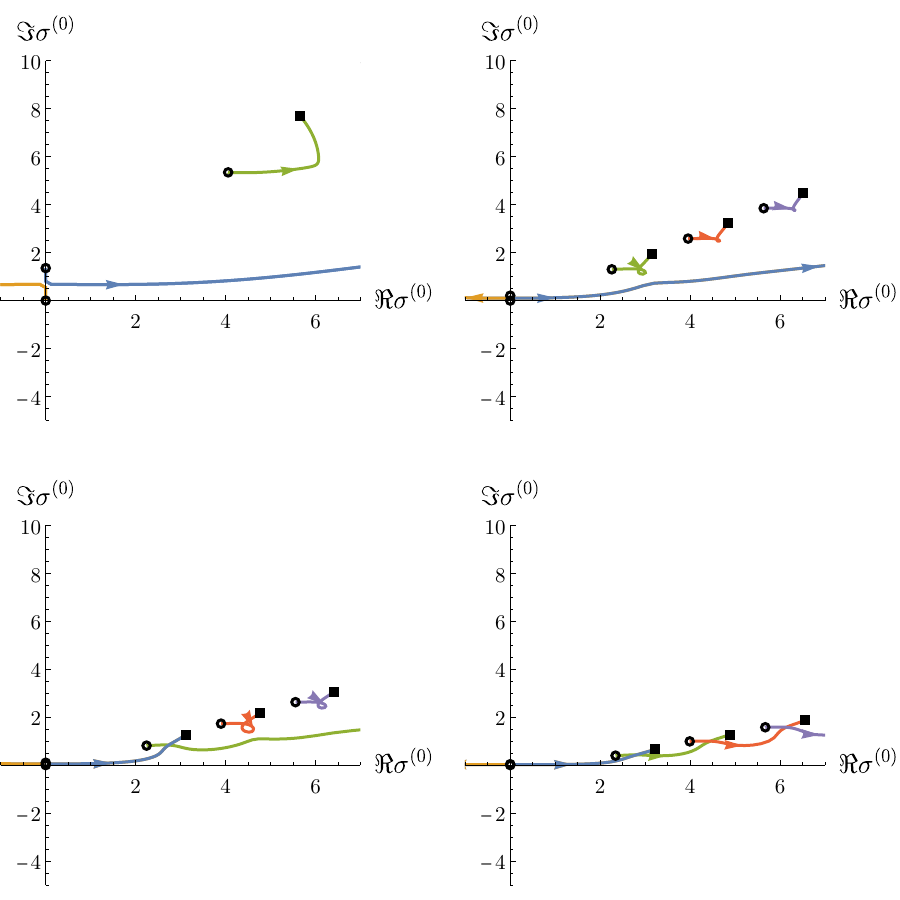}
 	\caption{The lowest QNMs of $\Phis^{(0)}$ (defocusing case) for fixed $\xh=1/2$, $2$, $3$, and $5$ (from left to right and from top to bottom) as a function of $b$. When $b$ decreases (indicated by arrows) from $\bstar(\xh)$ to $-\infty$ modes go from the spectrum of the trivial solution $\omega$ (circles) towards the spectrum of the singular solution $\sqrt{2}/y$ (squares), see \eqref{eq:16.11.23_1}. 	
 	}
	\label{fig:StaticConformalDefocusingQNM2}
\end{figure}

\section{Static solutions in SAdS ($\xh\rightarrow 0$)}
\label{eq:StaticSolutionsLargeBH}

In this section we discuss the limit $\xh\rightarrow 0$. Since static solutions with $b$ fixed and $\xh\rightarrow 0$ do not exist in the defocusing case, see Sec.~\ref{sec:StaticDefocusingCase}, here we consider only the focusing nonlinearity $\lambda=-1$. First, introduce rescaled radial coordinate
\begin{equation}
	\label{eq:23.10.25_03}
	z = 1- \frac{2}{\xh}y\,, \quad z\in [-1,1]
	\,, 
\end{equation}
and a new dependent variable
\begin{equation}
	\label{eq:23.10.25_04}
	S(z) := \xh\,\Phis\left(\frac{\xh}{2}(1-z)\right)
	\,.
\end{equation}
Then, making this change of variables in \eqref{eq:23.10.25_01} and expanding around $\xh=0$ we get in the leading order
\begin{equation}
	\label{eq:23.10.25_05}
-(1+z)\left(z^2-4z+7\right) S''(z)-3 (z-1)^2 S'(z)+(1-z) S(z) -2 S(z)^3 = 0
	\,.
\end{equation}
Note that under \eqref{eq:23.10.25_03} and \eqref{eq:23.10.25_04} the Robin BC with finite $b$ becomes the Neumann BC when $\xh\rightarrow 0$ 
\begin{equation}
	\label{eq:23.10.25_06}
	\left.\frac{\partial_{z}S}{S}\right|_{z=1} = -\frac{\xh}{2}b\ \xrightarrow{\xh\rightarrow 0}\ 0
	\,,
	\quad
	\textrm{for}
	\ 
	b<\infty
	\,.
\end{equation}
Thus, we look for regular solutions of \eqref{eq:23.10.25_05} with either the Dirichlet BC ($S(1)=0$) or the Neumann BC ($S'(1)=0$).
The shooting procedure is analogous to the $\xh>0$ case. It turns out that the equation \eqref{eq:23.10.25_05} has a countable family of solutions $S^{(n)}$, $n=0,1,\ldots$, both for $S(1)=0$ and $S'(1)=0$, see Fig.~\ref{fig:StaticConformalFocusingLimitingSolutions}, and each member of these families is a limiting solution for the rescaled solutions with $\xh>0$ as $\xh\rightarrow 0$ with respective boundary conditions. This fact is visualised in Fig.~\ref{fig:StaticConformalFocusingLimitingSolutionsConvergence}.

When trying to investigate the linear stability of these solutions with standard ansatz of the type \eqref{eq:23.10.26_01} one encounters problem when taking the limit $\xh\to 0$, namely, the term with supposed frequency $\sigma$ vanishes. This is consistent with the fact that for small $\xh$ the calculated values of $\sigma$ become very large, as can be seen in Fig.\ \ref{fig:StaticConformalFocusingQNM}. To deal with this problem, one can rescale the frequency $\Sigma=\xh\sigma$ used in the ansatz. Then the linear perturbation of \eqref{eq:23.10.25_05} is described by equation
\begin{equation}
	\label{eq:23.11.16_02}
8i\Sigma \chi' = -(1+z)\left(z^2-4z+7\right)\chi''-3(1-z)^2\chi'+(1-z-6S^2)\chi
\,.
\end{equation}
Its spectrum can be found with the shooting method as before, this time by constructing solutions starting at the singular point $z=-1$ and looking for $\Sigma$ for which they satisfy Dirichlet or Neumann BC at $z=1$. First of the BC is given as usual, by $\chi(1)=0$, while the second one is $2\chi'(1)+i\sigma\chi(1)=0$. The resulting values of $\Sigma$ for the node-less solution $S^{(0)}$ are presented in Table \ref{tab:xhzeroQNF}.

\begin{table}[t]
\centering
\begin{tabular}{|c|cccc|}
\toprule 
BC & $\Sigma^{(0)}_{0}$ & $\Sigma^{(0)}_{1}$ & $\Sigma^{(0)}_{2}$ & $\Sigma^{(0)}_{3}$ \\
\midrule
Neumann & $-0.67864i$ & $1.64333i$ & $\pm1.23933+2.71982i$ & $\pm2.81564+4.94251 i $ \\
Dirichlet & $-1.11063i$ & $ 2.89057i$  & $\pm2.28453 + 
 2.77389i$ & $\pm4.11190 + 4.93960i$  \\
 \bottomrule
\end{tabular}
\vskip 2ex
\caption{Lowest eigenvalues of linear perturbation of the static solution $S^{(0)}$ for Neumann and Dirichlet BC after rescaling by $\xh$. Frequencies of the unstable modes can be compared with Fig.~\ref{fig:StaticConformalFocusingQNMUnstable}.}
\label{tab:xhzeroQNF}
\end{table}

\begin{figure}[t]
	\centering
	\includegraphics[width=0.47\textwidth]{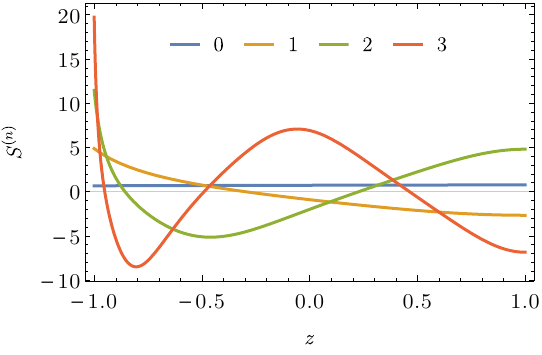}
	\hspace{4ex}
	\includegraphics[width=0.47\textwidth]{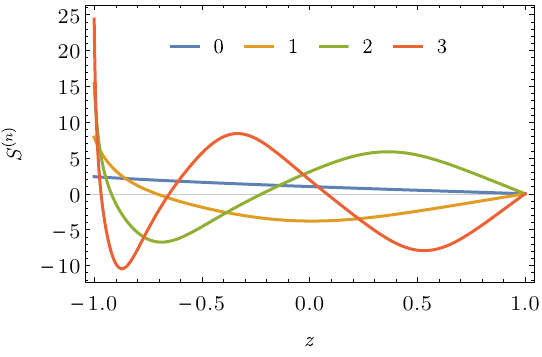}
	\caption{First four solutions $S^{(n)}$ of the limiting equation \eqref{eq:23.10.25_05}. Shown are solutions satisfying the Neumann (left plot) and Dirichlet (right plot) boundary conditions.}
	\label{fig:StaticConformalFocusingLimitingSolutions}
\end{figure}

\begin{figure}[t]
	\centering
	\includegraphics[width=0.47\textwidth]{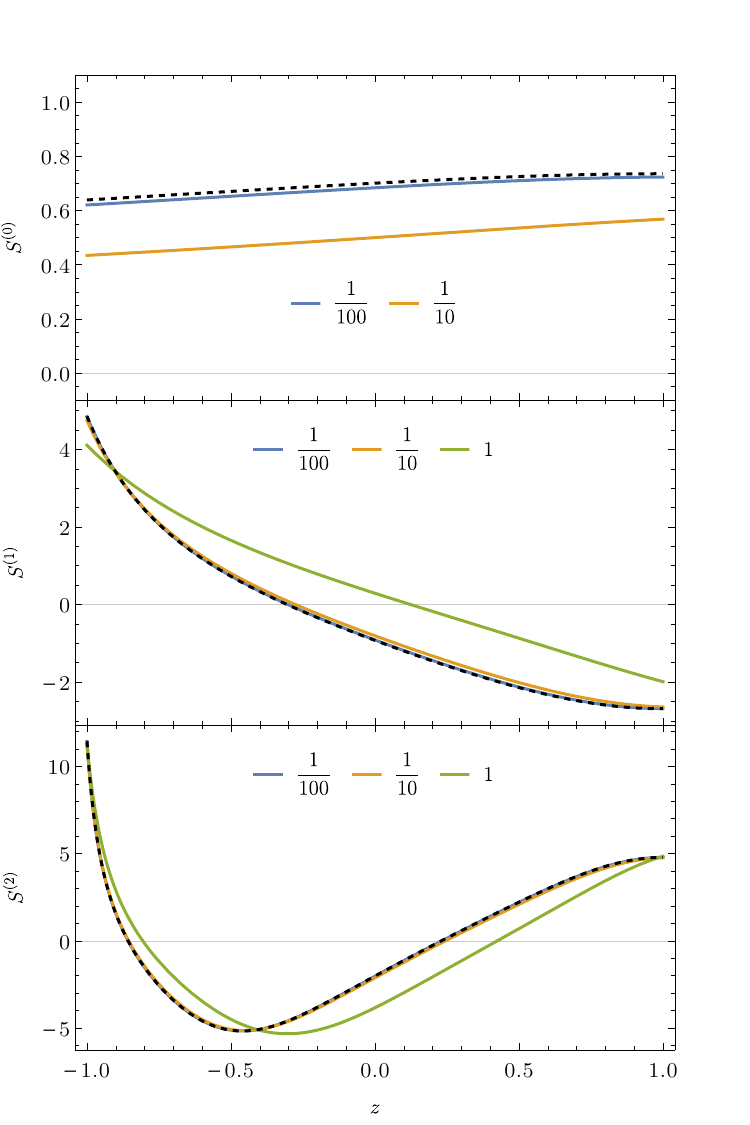}
	\hspace{4ex}
	\includegraphics[width=0.47\textwidth]{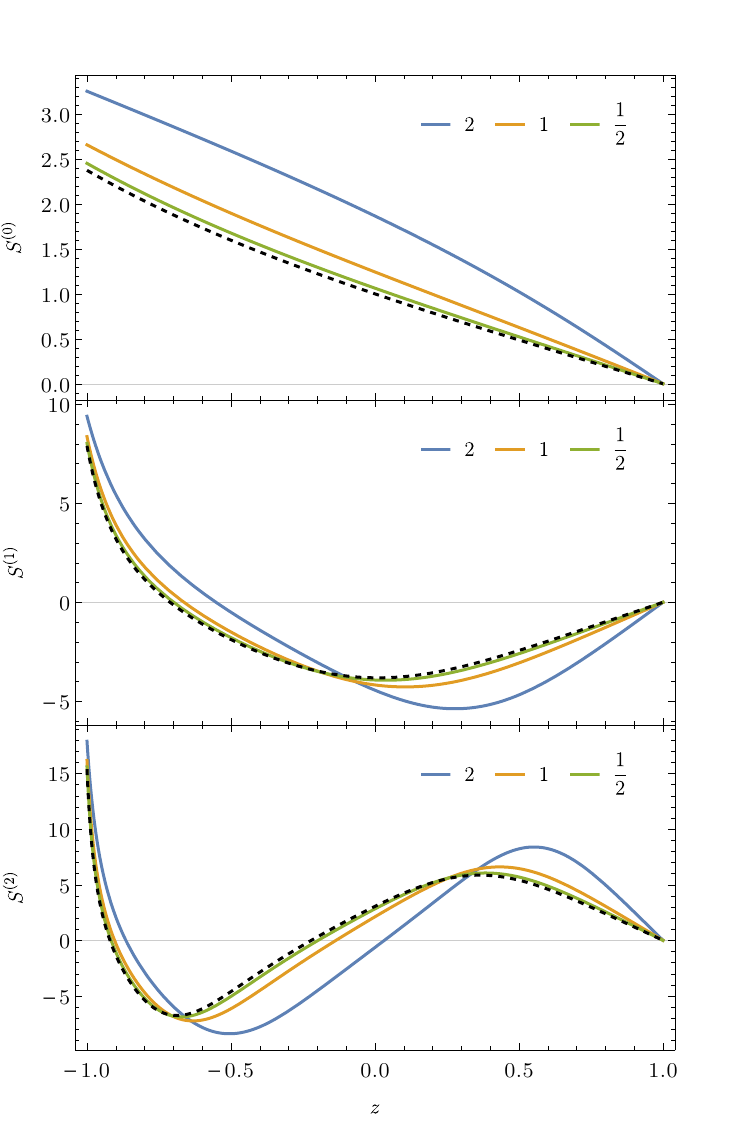}
	\caption{The convergence of rescaled solutions $\Phis\left(\frac{\xh}{2}(1-z)\right)\xh$ towards the respective solutions $S(z)^{(n)}$ of \eqref{eq:23.10.25_05} (dashed black line) in the limit $\xh\rightarrow 0$ (large BH). We show both the case with the Robin BC $b=-2$ (left plot) and the Dirichlet BC (right plot). The colour encodes the value of $\xh$.}
	\label{fig:StaticConformalFocusingLimitingSolutionsConvergence}
\end{figure}

\section{Dynamics ($0<\xh<\infty$)}
\label{sec:Dynamics}

\subsection{Focusing case $(\lambda=-1)$}
\label{sec:DynamicsFocusingCase}

As anticipated from the linear stability analysis the time evolution will depend on the values of $(\xh,b)$, in particular whether, for a given $\xh$, we take $b<\bstar(\xh)$ or $b>\bstar(\xh)$. However, within the respective regimes demarcated by the critical curve $\bstar(\xh)$ the dynamics is qualitatively the same. Below, we discuss the two cases $b<\bstar(\xh)$ and $b>\bstar(\xh)$ separately.

Before presenting the results, we describe the procedure used to solve the initial-boundary value problem \eqref{eq:22.09.22_09}-\eqref{eq:22.09.22_11} numerically.
For numerical convenience, we use the rescaling \eqref{eq:23.10.25_03} of the independent variable, which transforms \eqref{eq:22.09.22_09} into
\begin{equation}
	\label{eq:23.11.03_01}
	\partial_{z}\partial_{v}\Phi - \frac{\xh}{4}(1-z)\partial_{z}V(z)\Phi + \frac{\xh}{4}\partial_{z}\left((1-z)^{2}V(z)\partial_{z}\Phi\right) + \frac{2}{\xh(1-z)^{2}}\Phi - \lambda \frac{\xh}{4}\Phi^{3} = 0
	\,,
\end{equation}
where
\begin{equation}
	\label{eq:23.11.03_02}
	V(z) = 1 + \frac{4}{\xh^{2}(1-z)^{2}} - \frac{1}{2}\left(1+\frac{1}{\xh^{2}}\right)(1-z)
	\,,
\end{equation}
and the Robin BC \eqref{eq:22.09.22_11} becomes
\begin{equation}
	\label{eq:23.11.03_03}
	\left.\left(-\partial_{v}\Phi - \frac{2}{\xh}\partial_{z}\Phi - b\Phi\right)\right|_{z=1} = 0
	\,.
\end{equation}
To integrate the equations \eqref{eq:23.11.03_01}-\eqref{eq:23.11.03_03} we use the method of lines with standard 4th order Runge-Kutta time stepping and the Chebyshev pseudo-spectral discretisation in space. 
To get an explicit form of the evolution equations we solve \eqref{eq:23.11.03_01} for $\partial_{v}\Phi$ by inverting the $z$ derivative subject to the boundary condition \eqref{eq:23.11.03_03}. In practice, we use a Chebyshev grid of the second kind (Chebyshev-Gauss-Lobatto points), including boundary points $z=\pm 1$, and replace the equation at the $z=1$ grid point by the discrete version of \eqref{eq:23.11.03_03}. The inversion of the resulting square matrix is done by the LU decomposition algorithm.
The very same approach can be used when dealing with the Dirichlet BC $\Phi|_{z=1}=0$.

To get better results parts of the calculations are carried with extended numerical precision (typically 32-64 decimal digits).

We have experimented with several classes of initial conditions, but for all of them we got qualitatively the same outcomes. Below, we present numerical results which use
\begin{equation}
	\label{eq:23.11.03_04a}
	\Phi(0,z) = \amp\exp\left(-\frac{16}{(1-z^{2})^{2}}+16\right)
	\,,
\end{equation}
where $\amp\in\mathbb{R}$ is an adjustable parameter.

\subsubsection{$b>\bstar(\xh)$}
\label{sec:LargebDynamicsFocusingCase}

\begin{figure}[t]
	\centering
	\includegraphics[width=0.47\textwidth]{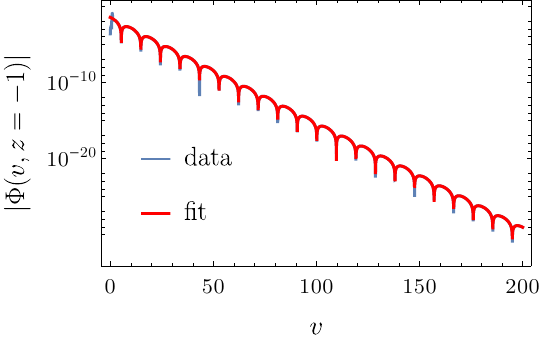}
	\hspace{4ex}
	\includegraphics[width=0.47\textwidth]{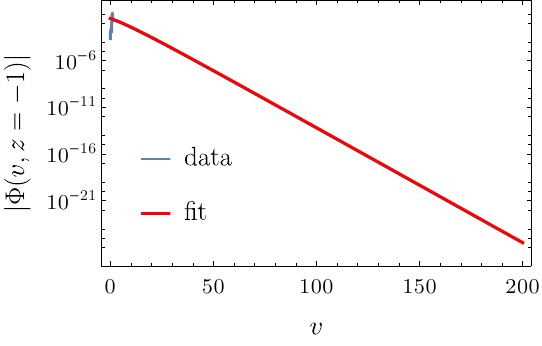}
	\caption{Evolution of small initial data, with amplitude $\amp=1/10$ in \eqref{eq:23.11.03_04a}, for $\xh=1$ and $b=-0.55$ (left plot) and $b=-0.61$ (right plot), for which we observe asymptotic decay to zero (we plot point-wise decay at the horizon $z=-1$). The late time dynamics is dominated by the lowest damped QNM, see Sec.~\ref{sec:LinearEquation}. Fits agree with the linear stability analysis of up to 10-12 significant digits. For the $b=-0.61$, which is below the bifurcation curve, see discussion in Sec.~\ref{sec:LinearEquation}, it was necessary to fit both purely imaginary modes which have similar exponents, while fitting just one exponential failed to produce correct results for the considered time interval.}
	\label{fig:EvolConformalFocusingSmallData}
\end{figure}

\begin{figure}[t]
	\centering
	\includegraphics[width=0.5\textwidth]{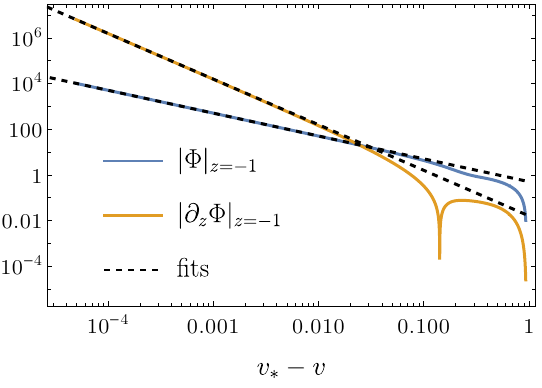}
	\caption{Finite-time blowup of large initial data ($\amp=3$ in \eqref{eq:23.11.03_04a}) observed at the horizon $z=-1$ (the $\xh=1$, $b=-0.55$ case). Plotted are the scalar field and its gradient as a function of the distance to the singularity $v_{*}\approx 0.944358$. These quantities grow as $(v_{*}-v)^{-1}$ and $(v_{*}-v)^{-2}$ respectively (dashed lines).}
	\label{fig:EvolutionConformalFocusingDetailsOfBlowup}
\end{figure}

For small initial perturbations (small $\amp$ in \eqref{eq:23.11.03_04a}) we observe convergence toward the zero solution. Depending on the combination of $(\xh,b)$ the exponential decay of the solution can be oscillatory or not. The late time behaviour is governed by the dominant QNM of $\Phi=0$ (these were discussed in Sec.~\ref{sec:LinearEquation}). This behaviour is illustrated in Fig.~\ref{fig:EvolConformalFocusingSmallData}, which also provides an independent verification of the computed QNMs in the linear stability analysis.

For large $\amp$'s, we see the solution blowing up in finite time. Growth is the fastest on the horizon, located at $z=-1$ on the numerical grid \eqref{eq:23.11.03_01}. 
The behaviour of solution as $v$ approaches the blowup time $v_{*}$ can be characterised by the following growth rates:
\begin{equation}
	\label{eq:23.11.03_04}
	\Phi(v,z=-1) \sim \frac{C_{0}}{v_{*}-v}
	\,, 
	\quad
	\partial_{z}\Phi(v,z=-1) \sim \frac{C_{1}}{(v_{*}-v)^{2}}
	\,,
\end{equation}
where numbers $C_{0},\ C_{1}\in\mathbb{R}$ are initial data dependent constants. Those scalings are shown in Fig.\ref{fig:EvolutionConformalFocusingDetailsOfBlowup}

\begin{figure}[t]
	\centering
	\includegraphics[width=1.0\textwidth]{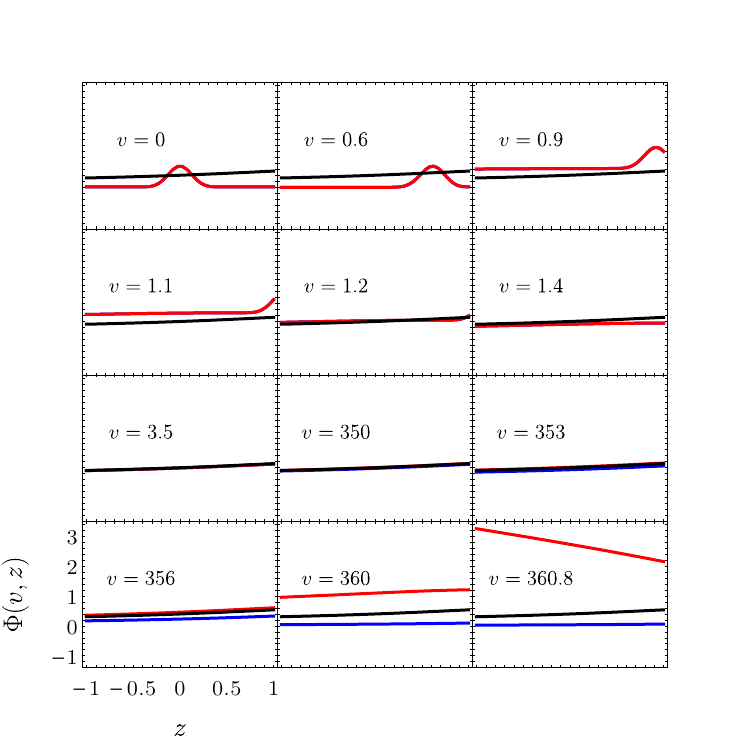}
	\caption{Time evolution of marginally sub- and supercritical data. The static solution (black) is an attractor for the finely tuned data for intermediate times. Subcritical (blue) and supercritical (red) data are indistinguishable up to $v\approx 350$, but later, they either converge towards zero or blow up to infinity, respectively. Here, we show the $\xh=1$, $b=-0.55$ case. The nearly critical evolution looks qualitatively the same for other parameters with $b>\bstar(\xh)$.}
	\label{fig:EvolutionConformalFocusingEvolutionSubSuperFrames}
\end{figure}

Performing a bisection between decay to zero and blowup we find that the static solution $\PhisZero$ (possessing a single unstable mode) acts as an intermediate attractor. Snapshots from the evolution of marginally super- and subcritical data are shown in Fig.~\ref{fig:EvolutionConformalFocusingEvolutionSubSuperFrames} for a representative case $\xh=1$, $b=-0.55$. Solutions differing by $\sim 10^{-64}$ in the amplitude of \eqref{eq:23.11.03_04} approach the static solution along its least damped QNM. For early and intermediate times the two curves are indistinguishable. Since the data is not exactly critical the unstable mode of $\PhisZero$ with $\Im\sigma^{(0)}_{0}<0$ becomes non-negligible and the solutions diverge along it.
This stage of the evolution can be described by the linearised solution:
\begin{equation}
	\label{eq:23.11.03_05}
	\Phi(v,z) \approx \PhisZero(z) + (\amp-\amp_{*})e^{i\sigma^{(0)}_{0}v}\chi_{0}(z) + \sum_{j>0}\alpha_{j}e^{i\sigma^{(0)}_{j}v}\chi_{j}(z)
	\,,
	\quad \alpha_{j}=\const
	\,,
\end{equation}
where the sum is over all damped QNMs, and $\amp_{*}$ denotes the critical value of the parameter of a chosen family of initial data \eqref{eq:23.11.03_04}.
Later, the supercritical data moves toward the finite-time blowup, which proceeds according to \eqref{eq:23.11.03_04}, while the subcritical data decays to zero following the QNMs of the trivial solution
\begin{equation}
	\label{eq:23.11.10_01}
	\Phi(v,z) \approx \beta_{0}e^{i\omega_{0}v}\psi_{0}(z) + \sum_{j>0}\beta_{j}e^{i\omega_{j}v}\psi_{j}(z)
	\,,
	\quad \beta_{j}=\const
	\,,
\end{equation}
with dominant contribution coming from the lowest damped mode $\omega_{0}$.
The details of this near-critical behaviour are illustrated in Fig.~\ref{fig:EvolutionConformalFocusingEvolutionSubSuperHorizon}.

\begin{figure}[t]
	\centering
	\includegraphics[width=0.47\textwidth]{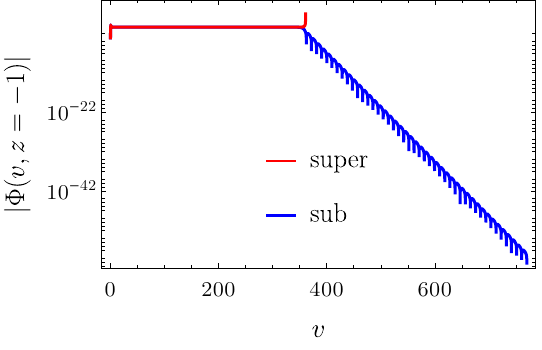}
	\hspace{4ex}
	\includegraphics[width=0.47\textwidth]{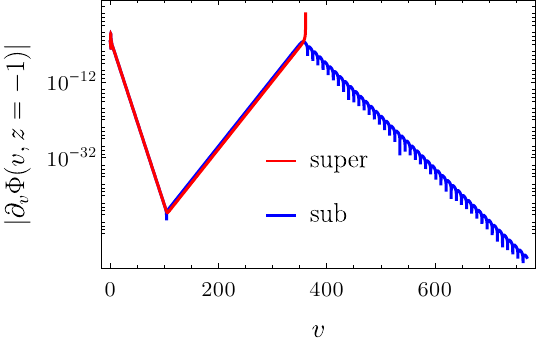}
	\caption{Time evolution of $|\Phi|$ and $|\partial_{v}\Phi|$ at the horizon for marginally sub- and supercritical data shown on Fig.~\ref{fig:EvolutionConformalFocusingEvolutionSubSuperFrames} (the $\xh=1$, $b=-0.55$ case). At intermediate times the solution is well approximated by the leading stable mode ($\sigma^{(0)}_{1}=\pm 1.26556 + 0.55092 i$) and the unstable mode ($\sigma^{(0)}_{-1}=-0.413592 i$) of the static solution $\PhisZero$. For late times the subcritical data follows the stable mode of the zero solution $\omega_{-1}=\pm 0.331294 + 0.317663 i$.}
	\label{fig:EvolutionConformalFocusingEvolutionSubSuperHorizon}
\end{figure}

This picture holds for any $0<\xh<\infty$ with $b>\bstar(\xh)$.
However, the quantitative behaviour of the nearly critical evolution will depend on the precise values of the parameters. The type of the decay is determined by mode with the smallest imaginary part, precisely whether it is a purely imaginary mode or a mode with nonzero real part. So in particular, for $b$ below the bifurcation curve the evolution of subcritical data is dominated by the non-oscillatory mode, see Fig.~\ref{fig:EvolutionConformalFocusingEvolutionSubSuperHorizonMultiple}.
As $b\rightarrow\bstar(\xh)$ the critical amplitude, $a_{*}$, tends to zero as well as the critical solution $s^{(0)}$. Thus, at $b=\bstar(\xh)$ we observe blowup for any data, see discussion below.

\begin{figure}[t]
	\centering
	\includegraphics[width=1.0\textwidth]{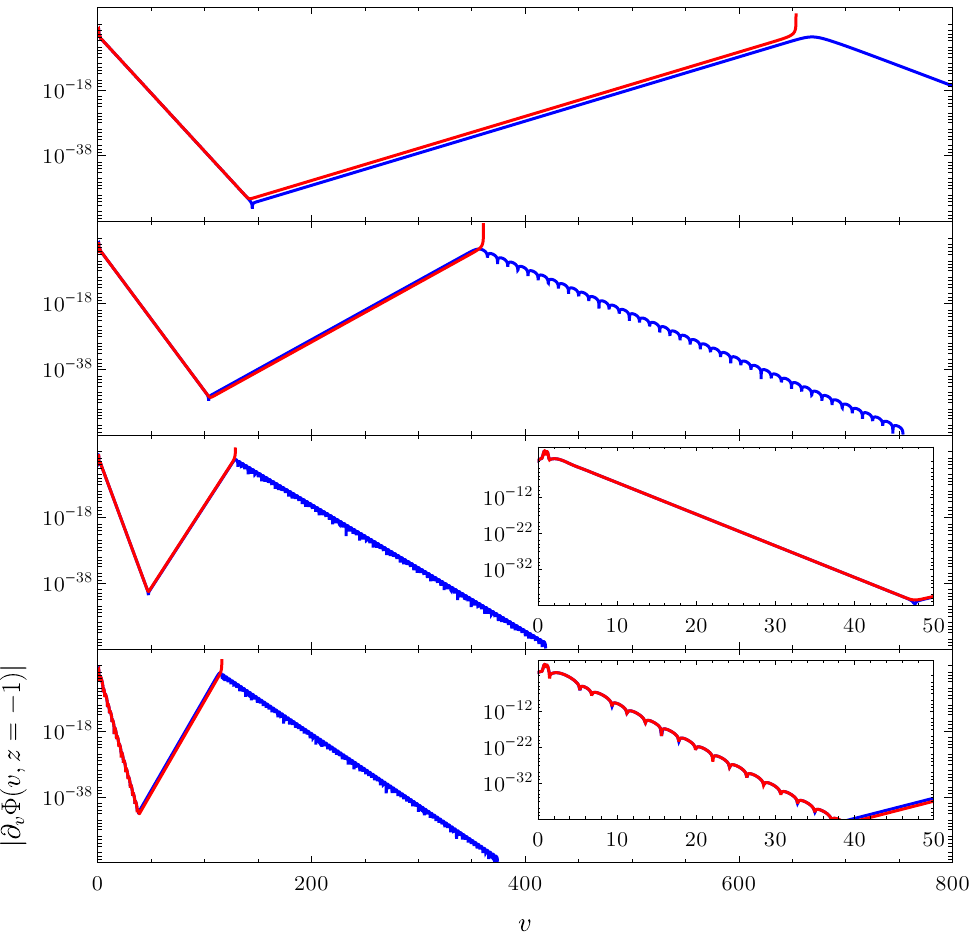}
	\caption{Analogue of Fig.~\ref{fig:EvolutionConformalFocusingEvolutionSubSuperHorizon} with $\xh=1$ for different values of $b$ (from top to bottom $b=-0.62, -0.55, 0, 0.25$). Depending on the values of $(\xh,b>\bstar(\xh))$ the decay towards the static solutions is oscillatory or not.}
	\label{fig:EvolutionConformalFocusingEvolutionSubSuperHorizonMultiple}
\end{figure}

\subsubsection{$b\leq\bstar(\xh)$}
\label{sec:SmallbDynamicsFocusingCase}

\begin{figure}[t]
	\centering
	\includegraphics[width=0.5\textwidth]{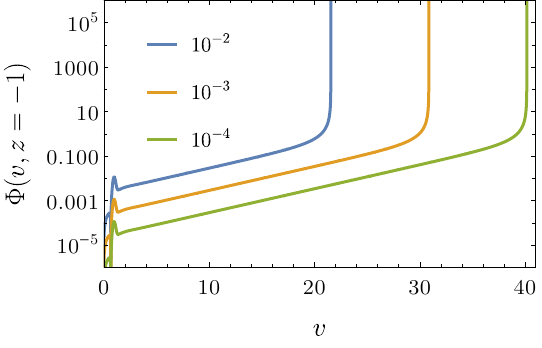}
	\caption{Typical evolution of small initial data for $b<\bstar(\xh)$ (point-wise behaviour of $\Phi$ at the black hole horizon $z=-1$). Data with different amplitudes $\amp$ of \eqref{eq:23.11.03_04} (colour coded) illustrate that before blow up, the solution is well approximated by the linearised solution \eqref{eq:23.11.03_06}. This is indicated by the slope of the curves, their relative vertical shift and the moment of onset of the nonlinear evolution. The initial growth rate agrees with the exponent of the unstable mode $\omega_{0}=-0.247504 i$ of the zero solution. Shown is the $(\xh,b)=(1,-0.75)$ case.}
	\label{fig:EvolutionConformalFocusingEvolutionUnderCriticalLineHorizon}
\end{figure}

In this case we observe blowup for any nonzero initial data. Starting with a very small perturbation of the trivial solution we observe that solution quickly approaches the spatial profile of the unstable mode $\psi_{0}$ of the zero solution, see Sec.~\ref{sec:LinearEquation}, and at the same time it exponentially grows with the exponent of that mode. Thus, during this phase of the evolution the solution is well approximated by the linearised solution:
\begin{equation}
	\label{eq:23.11.03_06}
	\Phi(v,z) \approx \psi_{0}(z) e^{i \omega_{0} v}
	\,.
\end{equation}
However, when the solution reaches certain threshold the nonlinearity becomes non-negligible and the nonlinear dynamics takes over. Consequently we observe a finite-time blowup \eqref{eq:23.11.03_04}. To further confirm that at early times the dynamics is indeed determined by the linear evolution \eqref{eq:23.11.03_06}, we show data with different amplitudes $\amp$ in \eqref{eq:23.11.03_04}, see Fig.~\ref{fig:EvolutionConformalFocusingEvolutionUnderCriticalLineHorizon}. The evolution looks qualitatively the same for any $(0<\xh<\infty,b<\bstar(\xh))$ including points on the critical line $\bstar(\xh)$.

\subsection{Defocusing case $(\lambda=1)$}
\label{sec:DynamicsDefocusingCase}

For defocusing nonlinearity solutions exist for all times irrespective of the size of initial data and for any choice of $(\xh,b)$. However, there is a qualitative change in the late time behaviour depending on whether we are above or below the critical curve $\bstar(\xh)$, at which we have a pitchfork bifurcation, with $\PhisZero$ bifurcating from zero.

In the region $b\geq\bstar(\xh)$ of the parameter space all solutions converge towards $\Phi=0$, as it is the only attractor, cf. Sec.~\ref{sec:StaticDefocusingCase}. After a series of nonlinear oscillations, the decay is dominated by the lowest damped QNM, whose frequency and damping rate depend on $(\xh,b)$, see Sec.~\ref{sec:LinearEquation}. Thus, the late time behaviour does not differ significantly from the small data case with focusing nonlinearity $\lambda=-1$.
\begin{figure}[t]
	\centering
	\includegraphics[width=0.47\textwidth]{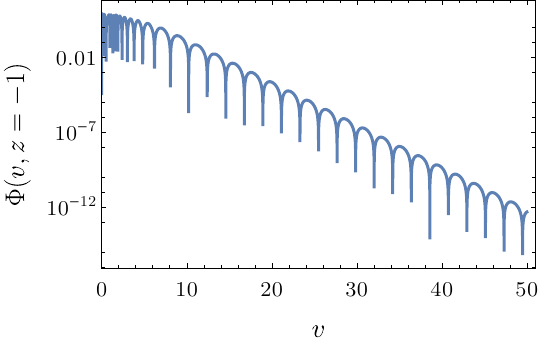}
	\hspace{4ex}
	\includegraphics[width=0.47\textwidth]{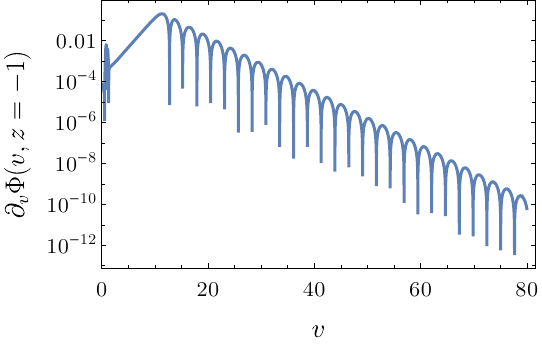}
	\caption{Illustration of qualitatively different behaviour for the cases $b>\bstar(\xh)$ (left plot) and $b<\bstar(\xh)$ (right plot) with $\xh=1$. For $b>\bstar(\xh)$, after a series of rapid nonlinear oscillations, an effect of large initial data, the solution converges to $\Phi=0$ via the dominant QNM. For $b<\bstar(\xh)$ even the tiniest perturbation of zero begins to grow exponentially ($\Phi=0$ is linearly unstable) before the solution settles on the static profile $\PhisZero$. For data comparison we plot $\Phi(v,z=-1)$ for $b=1$ (left plot) and $\partial_{v}\Phi(v,z=-1)$ for $b=-1$ (right plot).}
	\label{fig:EvolutionConformalDefocusing}
\end{figure}

However, if $b<\bstar(\xh)$ the zero solution is linearly unstable, Sec.~\ref{sec:LinearEquation}, and the linearly stable static solution $\PhisZero$ acts as a global attractor. Thus, for any initial data the solution settles on $\PhisZero$ and this approach is governed by the leading QNM $\sigma^{(0)}$. Details of dynamics of small initial data around zero is illustrated on Fig.~\ref{fig:EvolutionConformalDefocusing}. Interestingly, although the energy of the initial data is a fraction of the energy of the static solution $\PhisZero$, the perturbation grows in magnitude while being sourced from the boundary. In general, depending on the magnitude of initial data and the values of $(\xh,b)$, the energy falls into the horizon but can also be pumped into the domain by or radiated away through the conformal boundary, see Fig.~\ref{fig:EvolutionConformalDefocusingFluxes}.

\begin{figure}[t]
	\centering
	\includegraphics[width=0.47\textwidth]{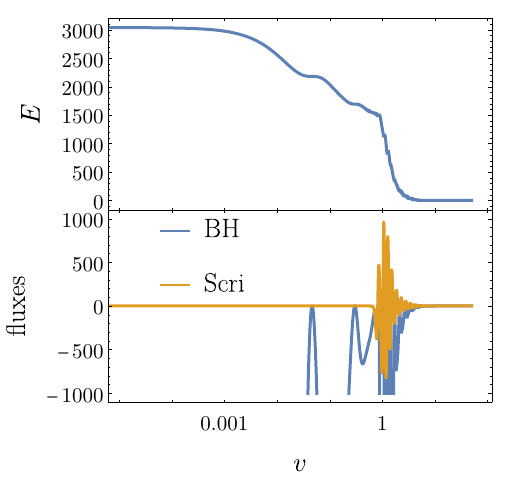}
	\hspace{4ex}
	\includegraphics[width=0.47\textwidth]{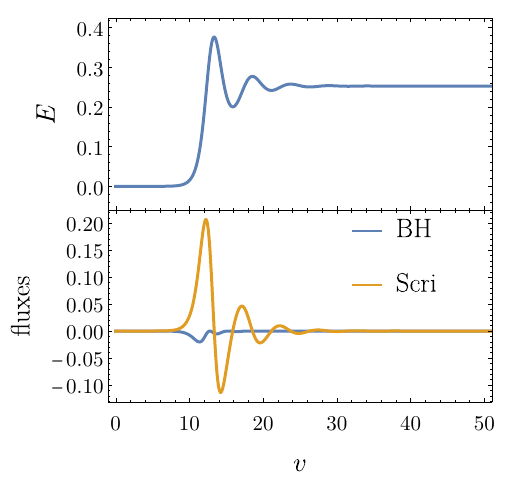}
	\caption{The energy \eqref{eq:23.10.03_01} and fluxes \eqref{eq:23.11.10_02} through BH horizon and Scri for solutions shown on Fig.~\ref{fig:EvolutionConformalDefocusing}; $(\xh,b)=(1,1)$ and $(\xh,b)=(1,-1)$ on left and right plots respectively.}
	\label{fig:EvolutionConformalDefocusingFluxes}
\end{figure}

\section{Conclusions}
\label{sec:Conclusions}

\begin{figure}[t]
	\centering
	\includegraphics[width=1.0\textwidth]{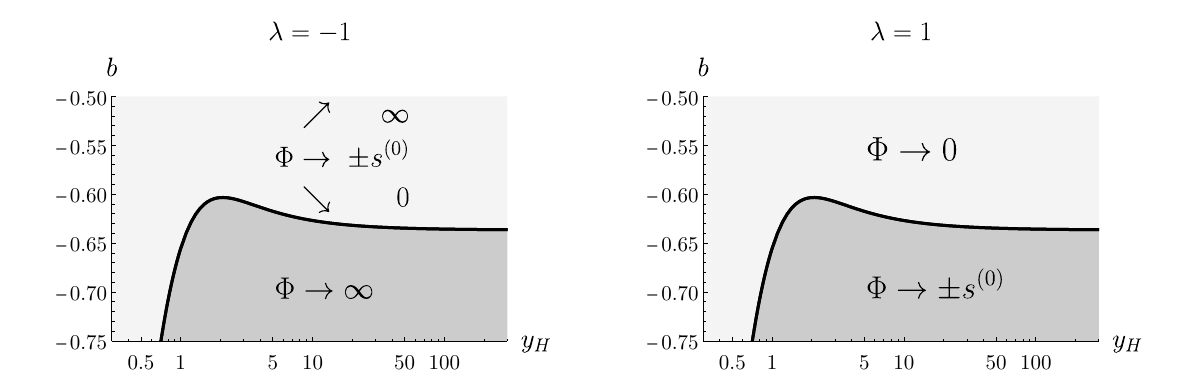}
 	\caption{Distinct asymptotic behaviours separated by the critical curve $\bstar(\xh)$ on the phase diagram $(\xh,b)$ for focusing (left plot) and defocusing (right plot) nonlinearities.}
	\label{fig:LastFigure}
\end{figure}

The main goal of this work was to understand the dynamics of nonlinear scalar field on SAdS background with Robin BC. We were in particular interested in how the behaviour of the field depends on the parameters of the model: the size of the black hole $\xh$ and the boundary condition $b$. It turns out that it is determined by the location of these parameters in the phase space with respect to the critical curve $\bstar(\xh)$, see Fig.~\ref{fig:LastFigure}.

In defocusing case for $(\xh,b)$ lying above the critical curve the zero solution plays the role of a global attractor: any initial data converges to it asymptotically. For $(\xh,b)$ below the critical curve this role is taken by two (up to the sign) static node-less solutions. These solutions bifurcate from zero, via pitchfork bifurcation, at the critical curve. The shape of this curve also leads to the conclusion that for sufficiently large black holes the zero solution is the attractor, which can be interpreted as the black hole absorbing the whole perturbation.

For focusing nonlinearity when parameters $(\xh,b)$ are below the critical curve one observes a nonlinear instability: all solutions eventually blow up at the horizon. On the other hand, for $(\xh,b)$ above this curve there exists a threshold, field configurations below it converge to zero, while above it one observes the blowup. In between this dichotomy lies $\PhisZero$ being a codimension-one attractor, as can be seen in Fig.~\ref{fig:EvolutionConformalFocusingEvolutionSubSuperHorizonMultiple}. As one gets closer to the critical curve in the parameter space this threshold decreases, eventually reaching zero. Additionally, for fixed $b$ sufficiently large black holes are stable for small initial data, as can be concluded from the behaviour of the critical curve for small $\xh$.

Let us briefly mention here that analogous behaviours can be observed in the massless case, i.e.\ when $m^2=0$ in \eqref{eq:22.09.22_04}. Then, as discussed in Sec.~\ref{sec:Introduction}, one is forced to assume the Dirichlet BC and the dynamics is qualitatively identical to the one observed for conformal equation with the same condition. In case of the defocusing nonlinearity, there are no static solutions and any field configuration converges to zero. For focusing nonlinearity, there exists a threshold separating initial data leading to a finite-time blowup and converging to zero.

To understand the dynamics of the considered system we also needed to study the static solutions, in particular their existence and linear stability (including the stability of zero solutions). It led us to a rather comprehensive grasp on their properties, as we discuss above, including limiting cases as $\xh\rightarrow 0$ and $\xh\rightarrow\infty$. The latter will lay foundation for our further research regarding dynamics of nonlinear scalar field on AdS background with Robin BC. In this case, due to the lack of the black hole, there is no simple mechanism of the energy loss for the system. It is possible that one observes there more complicated behaviour of the field, e.g., weak turbulence \cite{Bizoń.2011}.

Another potential direction that we would like to follow regards solutions that are axially symmetric. In this case, one can expect the logarithmic decay of the field \cite{Holzegel.2014}, suggesting the possibility of a weak turbulence in the presence of a black hole horizon. Finally, a natural continuation of this research is to investigate the dynamics for the self-gravitating case \cite{Bizoń.2020}, either with or without the self-interaction. We also plan to pursue this matter in the future.

\vspace{4ex}
\noindent\textit{Acknowledgement.} We are grateful to Piotr Bizo\'n and Claude Warnick for helpful remarks. We acknowledge the support of the Austrian Science Fund (FWF), Project \href{http://doi.org/10.55776/P36455}{P 36455} and the START-Project \href{http://doi.org/10.55776/Y963}{Y963}.

\printbibliography

\end{document}